\documentclass[11pt, letterpaper]{article}

\usepackage[margin=1.2in, top=1.1in, bottom=1.1in]{geometry} 
\usepackage[T1]{fontenc}
\usepackage[utf8]{inputenc}
\usepackage{tgpagella} 
\usepackage[varg]{newtxmath} 
\usepackage{microtype} 
\usepackage{authblk}

\usepackage{amsmath}
\usepackage{amsthm}
\usepackage{mathtools}
\usepackage{bm}
\usepackage{bbm}

\usepackage{graphicx}
\graphicspath{ {Figures/} }
\usepackage{subcaption}
\usepackage{booktabs}  
\usepackage{enumerate} 

\usepackage{algorithm}
\usepackage{algpseudocode}

\usepackage{xcolor}
\definecolor{DesignBlue}{RGB}{26, 54, 93}   
\definecolor{DesignGray}{RGB}{100, 110, 120} 

\usepackage{titlesec}
\titleformat{\section}{\large\bfseries\color{DesignBlue}}{\thesection.}{0.5em}{}
\titleformat{\subsection}{\normalsize\bfseries\color{DesignBlue}}{\thesubsection}{0.5em}{}
\titleformat{\subsubsection}{\normalsize\itshape\color{DesignBlue}}{\thesubsubsection}{0.5em}{}

\titlespacing*{\section}{0pt}{2.5ex plus 1ex minus .2ex}{1.5ex plus .2ex}

\usepackage{fancyhdr}
\usepackage{cite}
\usepackage[psdextra]{hyperref}
\hypersetup{
    colorlinks=true,
    linkcolor=DesignBlue,
    citecolor=DesignBlue,
    urlcolor=DesignBlue
}
\usepackage{cleveref}

\newtheoremstyle{modern}
  {12pt}
  {12pt}
  {}
  {}
  {\bfseries\color{DesignBlue}}
  {.}
  {.5em}
  {}
  
\newtheoremstyle{modernremark}
  {6pt}{6pt}{\normalfont}{}{\itshape\color{DesignBlue}}{.}{.5em}{}

\theoremstyle{modern}
\newtheorem{thm}{Theorem}
\newtheorem{prop}{Proposition}

\newtheorem{assume}{Assumption}

\theoremstyle{modernremark}
\newtheorem{rem}{Remark}
\newtheorem{ex}{Example}

\Crefname{thm}{thm}{Theorems}
\Crefname{prop}{Proposition}{Propositions}
\Crefname{assume}{assume}{Assumptions}
\Crefname{problem}{Problem}{Problems}
\Crefname{lem}{Lemma}{Lemmas}

\usepackage{tikz}
\usetikzlibrary{shapes, arrows, positioning}
\tikzstyle{block} = [draw, rectangle, fill=white, minimum height=2.5em, minimum width=4em, line width=0.8pt]
\tikzstyle{sum} = [draw, circle, inner sep=0pt, minimum size=6mm, line width=0.8pt]

\newcommand{\zono}[1]{\langle #1 \rangle}
\def \rad{\mathrm{ rad}}
\def \diag{\mathrm{ diag}}

\newcommand{\keywords}[1]{%
    \vspace{1em}
    \noindent\textbf{\color{DesignBlue}Keywords---}%
    \hspace{0.5em}\parbox[t]{\dimexpr\textwidth-7em}{\raggedright #1}
}

\title{\vspace{-2em}\huge\bfseries\color{DesignBlue}
Secure Set-based State Estimation for Safety-Critical Applications under Adversarial Attacks on Sensors
}

\author[1]{M.~Umar~B.~Niazi}
\author[2]{Michelle~S.~Chong}
\author[3]{Amr~Alanwar}
\author[1]{Karl~H.~Johansson}

\affil[1]{Digital Futures and Department of Decision and Control Systems, EECS, KTH Royal Institute of Technology, SE-100 44 Stockholm, Sweden. \textit{Emails}: \href{mailto:mubniazi@kth.se}{mubniazi@kth.se}, \href{mailto:kallej@kth.se}{kallej@kth.se}}
\affil[2]{Department of Mechanical Engineering, Eindhoven University of Technology, Eindhoven, The Netherlands. \textit{Email}: \href{mailto:m.s.t.chong@tue.nl}{m.s.t.chong@tue.nl}}
\affil[3]{School of Computation, Information and Technology, Technical University of Munich, Heilbronn, Germany. \textit{Email}: \href{mailto:alanwar@tum.de}{alanwar@tum.de}}

\date{}

\begin{document}
\maketitle
\thispagestyle{empty}

\begin{abstract}
    Set-based state estimation provides guaranteed state inclusion certificates that are crucial for the safety verification of dynamical systems. However, when system sensors are subject to cyberattacks, maintaining both safety and security guarantees becomes a fundamental challenge. Existing point-based secure state estimation methods cannot adequately address this challenge because they cannot provide state inclusion certificates. This paper introduces a novel approach that simultaneously ensures safety guarantees through guaranteed state inclusion and security guarantees against sensor attacks, without imposing conservative restrictions on system operation. We propose a Secure Set-based State Estimation (S3E) algorithm that maintains the true system state within the estimated set under sensor attacks, provided the initialization set contains the initial state and the system remains observable from the uncompromised sensor subset. The algorithm provides the estimated set as a collection of constrained zonotopes (agreement sets), which can be used as robust certificates to verify whether the system adheres to safety constraints. Furthermore, we demonstrate that the estimated set remains unaffected by attack signals of sufficiently large magnitude and also establish sufficient conditions for attack detection, identification, and filtering. This compels the attacker to inject only signals of small magnitudes to evade detection, thus preserving the accuracy of the estimated set. To address the computational complexity of the algorithm, we offer several strategies for complexity-performance tradeoffs. The efficacy of the proposed algorithm is illustrated through several examples, including its application to a three-story building model.
\end{abstract}


\clearpage

\section{Introduction}

Secure state estimation is critical for the reliable operation of cyber-physical systems, where control performance and safety depend directly on accurate state knowledge. Adversarial manipulation of sensor networks through false data injection can compromise system stability and lead to operational failures \cite{segovia2024}. This vulnerability presents significant challenges for critical infrastructure systems, including power grids, water distribution networks, autonomous vehicles, and industrial automation \cite{kayan2022}.

Modern control systems exhibit inherent vulnerabilities due to their distributed architectures and dependence on networked communication infrastructure. These characteristics enable adversaries to manipulate sensor measurements through false data injection attacks, creating divergence between reported and actual system states. Research demonstrates that sophisticated attacks can maintain stealth characteristics while progressively degrading system performance \cite{shoukry2018, yong2018}. The evolution of such attack strategies, coupled with increased system connectivity, drives the need for resilient estimation algorithms capable of maintaining state accuracy under sensor compromise.

Traditional approaches to secure state estimation have focused on point-based methods that produce single estimates of the system state at each time instant. These algorithms typically employ redundant sensors and leverage techniques such as majority voting or optimization-based outlier detection to identify and exclude corrupted measurements \cite{fawzi2014, chong2020, kim2018, pajic2016, he2021, alanwar2019, chang2018}. While these methods have shown promise in specific scenarios, they suffer from fundamental limitations that restrict their applicability. Most notably, point-based secure estimators require that strictly fewer than half of the sensors remain uncompromised. This constraint may be violated by adversaries with sufficient resources to coordinate large-scale attacks. Furthermore, these methods produce conservative error bounds that limit their practical applicability.

Recent efforts to address these limitations have explored the integration of cryptographic authentication mechanisms to verify the integrity of sensor data periodically \cite{lesi2017, khazraei2022}. However, the computational overhead associated with cryptographic operations creates challenging tradeoffs between security guarantees and the real-time performance requirements of control systems. This overhead becomes particularly problematic in resource-constrained embedded systems where computational capacity must be carefully allocated between control tasks and security operations.

\subsection{Set-based Estimation Framework}

Beyond merely detecting attacks or providing state estimates, safety-critical systems require formal verification that the system remains within safe operating bounds even under sensor attacks. Point-based methods fail to address this verification challenge as they cannot represent the complete space of feasible system states given measurement data and attack constraints without excessive conservatism. This limitation becomes critical in industrial processes operating with minimal safety margins, where minor estimation uncertainties may result in constraint violations.

Set-based estimation techniques, particularly those employing zonotopic representations, offer a fundamentally different approach to secure state estimation that addresses many limitations of point-based methods while naturally supporting safety verification \cite{alamo2005, le2013, scott2016, althoff2021, althoff2021-2, depaula2022, alanwar2023, rego2020}. These methods compute bounded sets containing all feasible system states rather than single estimates, yielding inclusion guarantees necessary for constraint verification. 
Moreover, set-based estimation methods naturally embed uncertainties and noise when their bounds are known. 
The zonotopic representation balances computational efficiency with tight uncertainty bounds, making real-time implementation feasible.

The key advantage of set-based methods for safety verification lies in their ability to propagate uncertainty through system dynamics while maintaining guaranteed bounds. These methods compute reachable sets encompassing all trajectories consistent with measurements and bounded disturbances, enabling verification of safety constraint satisfaction. This capability proves essential against stealthy attacks that evade immediate detection while gradually compromising system safety. Applications spanning fault diagnosis in industrial systems \cite{blesa2012}, underwater robotics \cite{jaulin2009}, vehicle localization \cite{bouron2001}, and water network leak detection \cite{rego2021} demonstrate the practical value of guaranteed uncertainty bounds for operational safety and performance verification.

\subsection{Current Limitations in Secure Set-Based Estimation}
Despite their theoretical advantages for safety verification, existing set-based methods for secure state estimation face practical challenges. The reachability-based approach of \cite{shinohara2018b} imposes that any subset of uncompromised sensors can measure the full state vector. Without such strong observability conditions, the estimation bounds become excessively conservative, potentially encompassing regions too large for effective control or meaningful safety verification. Similarly, 
the work \cite{niazi2023} on secure state estimation considered restrictive assumptions requiring observability from every individual sensor and was limited to systems that were bounded input-bounded state stable, which significantly constrains the class of systems that can be analyzed.

Alternative approaches that attempt to identify and filter compromised sensors before applying standard zonotopic filters \cite{meslem2020, chen2022, li2023} fail to address stealthy attacks. These attacks inject strategically designed signals that remain within expected bounds at each time step while causing significant cumulative deviation in state estimates. Although recent work has begun to address stealthy attacks \cite{liu2020, li2021, song2019, zhang2020, liu2021, zhu2022}, existing methods either restrict their analysis to specific attack strategies or assume known bounds on attack signal magnitudes. Sophisticated adversaries can readily violate these assumptions.

Furthermore, the fundamental assumption of requiring a majority of uncompromised sensors remains unaddressed in most set-based approaches. This assumption becomes increasingly problematic as systems scale and adversaries gain greater capabilities to compromise distributed sensor networks. These limitations necessitate novel algorithms that leverage the geometric properties of set-based representations to maintain state estimation accuracy and safety guarantees under extensive sensor compromise.

\subsection{Contributions of This Work}
This paper presents a Secure Set-Based State Estimation (S3E) algorithm that provides guaranteed state inclusion under adversarial sensor attacks while fundamentally relaxing the assumptions required by existing methods. Our approach leverages the redundancy inherent in zonotopic representations to maintain accurate state estimates and verify safety constraints even when all but one sensor can be compromised, requiring only that the system remains observable through any combination of remaining safe sensors. 

The S3E algorithm executes four sequential operations per time step. The time update computes feasible states based on system dynamics and bounded process noise. Measurement updates generate state sets consistent with sensor observations across different sensor subsets. Agreement sets form through strategic intersections of measurement updates, leveraging redundancy to exclude compromised sensor data. The final estimated set comprises the union of valid agreement sets, ensuring true state inclusion and enabling safety verification.

Our theoretical analysis establishes several key properties of the S3E algorithm with direct implications for safety verification. We derive explicit lower bounds on attack signal magnitudes that ensure detection at different stages of the estimation process, demonstrating that adversaries face fundamental constraints in designing attacks that compromise safety without detection. We characterize the conditions under which attacked sensors can be identified and filtered, providing system designers with quantitative metrics for evaluating security and safety margins. 
Additionally, we develop computational methods to verify estimated set containment within safety constraints, supporting real-time monitoring under active attacks.

\subsection{Outline of the Paper}
The rest of the paper is organized as follows. Section~\ref{sec:preliminaries} provides the required preliminaries, and Section~\ref{sec:prob} states the main assumptions and the problem. Section~\ref{sec:algo} presents the S3E algorithm, provides the inclusion guarantees, and discusses methods to reduce the algorithm's complexity. Section~\ref{sec:attack-detection-sec} provides lower bounds on the attacks that can be detected, identified, and/or filtered. 
Section~\ref{sec:evaluation} demonstrates the proposed algorithm through simulation examples. Finally, Section~\ref{sec:conclude} provides the concluding remarks.

\section{Notation and Preliminaries} \label{sec:preliminaries}

The sets of real numbers and integers are denoted by $\mathbb{R}$ and $\mathbb{Z}$, respectively. We let $\mathbb{Z}_{\geq i} \coloneq  \{i,i+1,i+2,\dots\}$ and $\mathbb{Z}_{[i,j]}\coloneq \{i,i+1,i+2,\dots,j\}$ for $j\geq i$. The Euclidean and maximum norms of a vector $x \in \mathbb{R}^{n}$ are denoted as $\|x\|$ and $\|x\|_\infty$, respectively. The vector and matrix of zeros are denoted as $0_n\in\mathbb{R}^n$ and  $0_{n\times k}\in\mathbb{R}^{n\times k}$. The identity matrix is $I_n \in \mathbb{R}^{n\times n}$.
For a finite set $\mathcal{S}$, $|\mathcal{S}|$ denotes its cardinality. 
The Cartesian product between two sets $\mathcal{S}_1$ and $\mathcal{S}_2$ is denoted by $\mathcal{S}_1 \times \mathcal{S}_2$. 

Given a center $c_{z}\in\mathbb{R}^n$ and a generator $G_z\in\mathbb{R}^{n\times \xi_z}$, a \textit{zonotope} $\mathcal{Z}=\zono{c_z,G_z}$ is a set
\[
\mathcal{Z}\coloneq  \{c_z+G_z\beta_z : \beta_z\in[-1,1]^{\xi_z}\}
\]
where $\xi_z$ is the number of generators of $\mathcal{Z}$. That is, a zonotope is an affine transformation of a unit hypercube $\mathcal{H}(0_n,1) = [-1,1]^{\xi_z}$ centered at $0_n$ and with radius $1$, where $\xi_z\in \mathbb{Z}_{\geq 1}$ is the dimension of the hypercube. 

A matrix $L\in\mathbb{R}^{n'\times n}$ multiplied with a zonotope $\mathcal{Z}$ yields a linearly transformed zonotope $L\mathcal{Z}=\zono{Lc_z,LG_z}$. Given two zonotopes $\mathcal{Z}_1=\zono{c_{z_1},G_{z_1}}$ and $\mathcal{Z}_2=\zono{c_{z_2},G_{z_2}}$ in $\mathbb{R}^n$, their Minkowski sum is given by
\[
\mathcal{Z}_1\oplus \mathcal{Z}_2 = \zono{c_{z_1}+c_{z_2},[\arraycolsep=2pt\begin{array}{cc} G_{z_1} & G_{z_2} \end{array}]}.
\]
The Cartesian product of two zonotopes is given by
\begin{align*}
    \mathcal{Z}_1 \times \mathcal{Z}_2 
    &\coloneq  \left\{ {\left[ {\begin{array}{c}
    {{z_1}}\\
    {{z_2}}
    \end{array}} \right]} : {z_1} \in {\mathcal{Z}_1},{z_2} \in {\mathcal{Z}_2} \right\} 
    = \left\langle \left[ \begin{array}{c}
    c_{z_1}\\ 
    c_{z_2}
    \end{array} \right],\left[ \begin{array}{cc}
    G_{z_1} & 0\\
    0 & G_{z_2}
    \end{array} \right] \right\rangle.
\end{align*}

A \textit{constrained zonotope} $\mathcal{Z}=\zono{c_z,G_z,A_z,b_z}$ is a set
\[
\mathcal{Z} \coloneq  \{c_z+G_z\beta_z: \beta_z\in[-1,1]^{\xi_z}, A_z\beta_z = b_z\}
\]
where $c_z\in \mathbb{R}^n$, $G_z \in \mathbb{R}^{n\times \xi_z}$, $A_z\in\mathbb{R}^{m \times \xi_z}$ and $b_z\in\mathbb{R}^m$ with $m,n\in\mathbb{Z}_{\geq 1}$. 
That is, a constrained zonotope is an affine transformation of the linearly constrained unit hypercube $\{\beta_z \in [-1,1]^{\xi_z} : A_z \beta_z = b_z \}$. 

Given two constrained zonotopes $\mathcal{Z} \subset \mathbb{R}^{n_z}$ and $\mathcal{Y} \subset \mathbb{R}^{n_y}$, and a matrix $M \in \mathbb{R}^{n_y \times n_z}$, the \textit{generalized intersection} is defined as
\begin{equation}
    \mathcal{Z} \cap_{M} \mathcal{Y}  := \{z \in \mathcal{Z}: M z \in \mathcal{Y}\}.
\end{equation}
Let $\mathcal{Z}=\zono{c_z,G_z,  A_z, b_z}\subset \mathbb{R}^{n_z}$ with $c_z\in \mathbb{R}^{n_z}$, $G_z\in \mathbb{R}^{n_z \times \xi_z}$, $A_z\in \mathbb{R}^{m_z \times \xi_z}$, and $b_z\in \mathbb{R}^{m_z}$, and let $\mathcal{Y}=\zono{c_y,G_y,  A_y, b_y} \subset \mathbb{R}^{n_y}$ with $c_y\in \mathbb{R}^{n_y}$, $G_y\in \mathbb{R}^{n_y \times \xi_y}$, $A_y\in \mathbb{R}^{m_y \times \xi_y}$, and $b_y\in \mathbb{R}^{m_y}$, then the generalized intersection can be computed as, see \cite{scott2016},
\begin{equation}
    \label{eq:gen_int_comp}
    \mathcal{Z} \cap_{M} \mathcal{Y} = 
    \Big\langle c_z,
    \left[\begin{array}{cc} G_z & 0_{n_z \times \xi_y} \end{array}\right],
    \left[\begin{array}{cc}A_z & 0_{m_z \times \xi_y} \\ 0_{m_y \times \xi_z} & A_y \\ M G_z & -G_y\end{array}\right],
    \left[\begin{array}{c}b_z \\ b_y \\ c_y-M c_z\end{array}\right]  \Big\rangle . 
\end{equation}

The radius of a zonotope, or a constrained zonotope, $\mathcal{Z}\subset\mathbb{R}^n$ is defined as
\[
\rad(\mathcal{Z}) \coloneq \min \Delta ~\text{subject to}~ \mathcal{Z}\subseteq\mathcal{H}(c_z,\Delta)
\]
where $\Delta$ is the radius of a minimal hypercube of dimension $n$ denoted by $\mathcal{H}(c_z, \Delta) \coloneq \zono{c_z,\Delta I_n}$,
which is centered at $c_z$ and inscribes $\mathcal{Z}$. Notice that for any point $p\in \mathcal{Z}$, it holds that
\[
\|c_z - p\| \leq \sqrt{n}\, \rad(\mathcal{Z})
\]
where $\rad(\mathcal{Z})=\sup_{p\in \mathcal{Z}} \|c_z-p\|_\infty$.

\section{Problem Definition} \label{sec:prob}

Consider a linear time-invariant (LTI) system
\begin{subequations}
\label{eq:system}
\begin{align}
    x(k+1) &= Ax(k) + Bu(k) + w(k)  \label{eq:system_state} \\
    y_i(k) &= C_ix(k) + v_i(k) + a_i(k), \quad i\in \mathbb{Z}_{[1,p]}\label{eq:system_output}
\end{align}
\end{subequations}
where $x(k)\in\mathbb{R}^{n_x}$ is the state, $u(k)\in\mathbb{R}^{n_u}$ is a known input, and $y_i(k)\in\mathbb{R}^{m_i}$ is the measured output of the $i$-th sensor with $i\in \mathbb{Z}_{[1,p]}$.
The vector $w(k)\in\mathcal{W}$ represents the process noise, which is assumed to be contained in the zonotope $\mathcal{W}=\zono{c_w,G_w}$, and the vector $v_i(k)\in\mathcal{V}_i$ represents the measurement noise of the $i$-th sensor, which is assumed to be contained in the zonotope $\mathcal{V}_i=\zono{c_{v_i},G_{v_i}}$, for every $i\in\mathbb{Z}_{[1,p]}$. Sensor~$i$'s measurement at time~$k$ may be corrupted by an arbitrary and unbounded attack signal $a_i(k)\in\mathbb{R}^{m_i}$, which is assumed to be designed by a malicious attacker.

\begin{assume} \label{assum:main}
We assume the following:
\begin{enumerate}[(i)]
    \item \label{assum:main_numAttacks} 
    \textbf{Maximum number of attacked sensors}: The attacker can attack up to $q \leq p-1$ number of sensors. Although the upper bound $q$ is known, the exact number and the set of attacked sensors do not need to be known.
    
    \item \label{assum:main_observability}
    \textbf{Redundant observability}: There exists $c_\mathsf{J}\leq p-q$ such that, for every subset $\mathsf{J}\subset\mathbb{Z}_{[1,p]}$ of sensors with cardinality $|\mathsf{J}|=c_\mathsf{J}$, the pair $(A,C_{\mathsf{J}})$ is observable, i.e., 
    $$\mathrm{rank}[\arraycolsep=2pt\begin{array}{cccc} C_\mathsf{J}^\top & (C_\mathsf{J} A)^\top & \dots & (C_\mathsf{J} A^{n_x-1})^\top \end{array}]^\top = n_x$$
    where $C_{\mathsf{J}}$ is obtained by stacking all $C_i$, $i\in \mathsf{J}$, in row blocks.

    \item \label{assum:initial_set}
    \textbf{Knowledge of an initial set}: At time $k=0$, we know a bounded set $\mathcal{X}_0\subset \mathbb{R}^{n_x}$ containing the initial state $x(0)\in \mathcal{X}_0$.

    \item \label{assum:noise_bounds}
    \textbf{Bounded noise}: The process and measurement noise zonotopes $\mathcal{W}, \mathcal{V}_1,\dots,\mathcal{V}_p$ are known.
\end{enumerate}
\end{assume}

Assumption~\ref{assum:main}(\ref{assum:main_numAttacks}) restricts the maximum number of attacked sensors at each time instant so that at least one sensor is uncompromised. This assumption is fundamental in this paper, ensuring that, at every time $k\in\mathbb{Z}_{\geq 0}$, there exists a set of \textit{uncompromised} (or \textit{safe}) sensors given by
\[
\mathsf{S}_k\subset\mathbb{Z}_{[1,p]} ~\text{with}~ |\mathsf{S}_k|\geq p-q
\]
such that $a_i(k)= 0_{m_i}$ for every $i\in\mathsf{S}_k$. This, along with the redundant observability (Assumption~\ref{assum:main}(\ref{assum:main_observability})), allows us to ensure that the true state can be theoretically reconstructed from the set of uncompromised sensors in the absence of noise. In addition, the assumption entails that the attacker, even though omniscient about the system dynamics and noise bounds, has limited resources at hand and cannot attack all the sensors. We remark that this assumption is not restrictive because it neither restricts the set of attacked sensors to be static with respect to time nor requires that $q$ is less than half the number of sensors $p$, which is the fundamental assumption in the point-based secure state estimation literature. On the contrary, at any time instant, our problem setup allows the attacker to inject arbitrary signals to any subset of sensors with cardinality less than or equal to $q$, where $q$ is only required to be strictly less than $p$.

Assumption~\ref{assum:main}(\ref{assum:initial_set}) is required to initialize the set-based state estimation algorithm with $\hat{\mathcal{X}}_0 = \mathcal{X}_0$ such that the inclusion of the state $x(k)$ in the estimated set $\hat{\mathcal{X}}_k\subset \mathbb{R}^{n_x}$ can be guaranteed. Note that Assumption~(\ref{assum:initial_set}) does not require any bound on the radius of $\mathcal{X}_0$. 
Finally, Assumption~\ref{assum:main}(\ref{assum:noise_bounds}) is a standard assumption in robust estimation and control \cite{zhou1998, freeman2008}.

\subsection*{Problem Statement}
Given the uncertain system \eqref{eq:system} subject to Assumption~\ref{assum:main}, we study the following problems:
\begin{enumerate}[i.]
    \item Estimate a set $\hat{\mathcal{X}}_k$ guaranteeing the inclusion of the true state, i.e., $x(k)\in\hat{\mathcal{X}}_k$ for every $k\in\mathbb{Z}_{\geq 1}$.
    \item Provide conditions such that the attack signals $a_i(k)$ can be detected, identified, and/or filtered\footnote{These terms are defined later in Section~\ref{sec:attack-detection-sec}.}.
\end{enumerate}

\section{Secure Set-based State Estimation Algorithm} \label{sec:algo}

This section presents the secure set-based state estimation (S3E) algorithm, which is summarized below in Algorithm~\ref{alg:s3e}. In the following subsections, we describe all the steps of the algorithm in detail, provide a comparison with point-based estimators, and present strategies to manage the algorithm complexity.

\begin{algorithm}[!]
\caption{Secure set-based state estimation}\label{alg:s3e}
\begin{algorithmic}[9]
\State Initialize: $\hat{\mathcal{X}}_{0} = \mathcal{X}_0$

\For {$k=1,2,3,\dots$}

\State Time update \eqref{eq:time_update}: $\hat{\mathcal{X}}_{k|k-1}$

\State Measurement update \eqref{eq:meas_update_i}: $\hat{\mathcal{Z}}_k^i$ for every $i\in\mathbb{Z}_{[1,p]}$.

\State Agreement protocol \eqref{eq:agreement_protocol}: $\mathcal{I}_k^h$ for $h\in \mathbb{Z}_{[1,n_\mathsf{J}]}$.

\State Estimate \eqref{eq:estimated_set}: $\hat{\mathcal{X}}_k$
\label{step:estimate}

\EndFor
\end{algorithmic}
\end{algorithm}

\subsection{Time Update}
\label{subset_timeupdate}

The first step in the S3E algorithm is the time update
\begin{equation}
\label{eq:time_update}
    \hat{\mathcal{X}}_{k|k-1} = A\hat{\mathcal{X}}_{k-1} \oplus Bu(k-1) \oplus \mathcal{W}
\end{equation}
where $\mathcal{W}$ is the zonotope that bounds the process noise $w(k)$, and $\hat{\mathcal{X}}_{k-1}$ is the estimated set computed in the previous time step. The computation of $\hat{\mathcal{X}}_k$ is described in Section~\ref{subsec_estimatedSet} with the initialization $\hat{\mathcal{X}}_0 = \mathcal{X}_0$, where $\mathcal{X}_0$ is the known initial set from Assumption~\ref{assum:main}(\ref{assum:initial_set}).

Given $x(k-1)\in\hat{\mathcal{X}}_{k-1}$, the time update set $\hat{\mathcal{X}}_{k|k-1}\ni x(k)$ is the set of states to which the system can evolve at time $k\in\mathbb{Z}_{\geq 0}$ subject to the model $(A,B)$, the input $u(k-1)$, and the noise zonotope $\mathcal{W}$.
Although the attacker cannot directly influence the time update set as seen from \eqref{eq:time_update}, it can do so indirectly through the previous estimate $\hat{\mathcal{X}}_{k-1}$. Thus, the resilience against sensor attacks is achieved by carefully devising the estimated set $\hat{\mathcal{X}}_{k}$.

\subsection{Measurement Update}
\label{subsec_measurementUpdate}

The time update \eqref{eq:time_update} uses model-based information to compute a set $\hat{\mathcal{X}}_{k|k-1}$ containing all points that the system could reach in one time step if they are initialized at the previous estimated set $\hat{\mathcal{X}}_{k-1}$. The measurement update, on the other hand, corrects the conservative estimate of the model-based time update by restricting it according to the sensor measurements.

For every $i\in\mathbb{Z}_{[1,p]}$, let
\begin{equation}
    \label{eq:output_sets}
    \mathcal{Y}_k^i \coloneq y_i(k) - \mathcal{V}_i = \zono{y_i(k)-c_{v_i}, G_{v_i}}
\end{equation}
be sensor~$i$'s output measurement set.
Then, the measurement update of sensor~$i$ is given by the following generalized intersection of the time update set $\hat{\mathcal{X}}_{k|k-1}$ with the output measurement set $\mathcal{Y}_{k}^{i}$ of sensor~$i$
\begin{align}
    \hat{\mathcal{Z}}_k^i &:= \hat{\mathcal{X}}_{k|k-1} \cap_{C_i} \mathcal{Y}_k^i = \{ x\in \hat{\mathcal{X}}_{k|k-1} : C_i \hat{\mathcal{X}}_{k|k-1} \in \mathcal{Y}_k^i \}
    \label{eq:meas_update_i}
\end{align}
which can be computed using \eqref{eq:gen_int_comp}.

The measurement update $\hat{\mathcal{Z}}_k^i$ is a subset of the time update set $\hat{\mathcal{X}}_{k|k-1}$ that is consistent with the measurement $y_i(k)$ of sensor~$i$. That is, given the measurement noise zonotope $\mathcal{V}_i$, $\hat{\mathcal{Z}}_k^i$ is the set of states $x(k)\in \hat{\mathcal{X}}_{k|k-1}$ that could have produced the measurement $y_i(k)$. Notice that $\hat{\mathcal{Z}}_k^i$ can be an empty set if the sensor~$i$'s measurement is inconsistent with the time update.

\subsection{Set-based Agreement Protocol}

Let $n_\mathsf{J} := {p \choose c_\mathsf{J}}$, where $c_\mathsf{J}\leq p-q$ is the redundant observability parameter. Let $\mathsf{J}_1, \dots, \mathsf{J}_{n_\mathsf{J}}\subset \mathbb{Z}_{[1,p]}$ be the  disjoint subsets with $|\mathsf{J}_h|=c_\mathsf{J}$, for $h\in\mathbb{Z}_{[1,n_\mathsf{J}]}$. 
The agreement protocol between measurement updates in each $\mathsf{J}_h$ is given by
\begin{equation}
    \label{eq:agreement_protocol}
    \mathcal{I}_k^h := \bigcap_{j\in \mathsf{J}_h} \hat{\mathcal{Z}}_k^j.
\end{equation}

\begin{thm}
\label{thm:safe-agreement}
Let Assumption~\ref{assum:main} hold and assume $x(k)\in \hat{\mathcal{X}}_{k|k-1}$. Then, at every time~$k$, there exists $h\in \mathbb{Z}_{[1,n_\mathsf{J}]}$ such that the agreement set $\mathcal{I}_k^h$ is non-empty and contains the true state $x(k)$.
\end{thm}
\begin{proof} 
We first show that the measurement update $\hat{\mathcal{Z}}_k^i$ contains the true state $x(k)$ if sensor~$i$ is in the uncompromised set of sensors $\mathsf{S}_k$ at time~$k$. If $i$ is uncompromised at time~$k$, then $a_i(k)=0$ and $y_i(k) = C_i x(k) + v_i(k)$.  From \eqref{eq:output_sets}, we have $\mathcal{Y}_k^i=\zono{c_{y_i(k)},G_{v_i}}$ with
$
c_{y_i(k)} = y_i(k) - c_{v_i} = C_i x(k) + v_i(k) - c_{v_i}.
$
As $\mathcal{Y}_k^i$ and the noise zonotope $\mathcal{V}_i$ have the same generator matrix $G_{v_i}$, we have $\rad(\mathcal{Y}_k^i) = \rad(\mathcal{V}_i)$. Moreover, it holds that
\[
\mathcal{Y}_k^i \supseteq \{c_{y_i(k)} - \xi : \|\xi\|_\infty \leq \rad(\mathcal{Y}_k^i) \}.
\]
Thus, 
$
c_{y_i(k)} - v_i(k) + c_{v_i} = C_i x(k) \in \mathcal{Y}_k^i
$
because $\|v_i(k) - c_{v_i}\|_\infty \leq \rad(\mathcal{V}_i) = \rad(\mathcal{Y}_k^i)$.
Since $x(k)\in \hat{\mathcal{X}}_{k|k-1}$ and $C_i x(k) \in \mathcal{Y}_k^i$, it implies that $\hat{\mathcal{Z}}_k^i$ in \eqref{eq:meas_update_i} is non-empty and contains $x(k)$.

We then show that, at every time~$k$, there is a subset $\mathsf{J}_h$ of sensors with $|\mathsf{J}_h| = c_\mathsf{J}$ such that $\mathsf{J}_h \subseteq \mathsf{S}_k$. 
This is straightforward because $q$ out of $p$ sensors are compromised and $c_\mathsf{J}\leq p-q$, so there must be at least one subset $\mathsf{J}_h$ with $|\mathsf{J}_h|=c_\mathsf{J}$ that contains the uncompromised sensors from the guaranteed $p-q$ uncompromised sensors. In other words, at every time $k\in \mathbb{Z}_{\geq 0}$, there exists $h\in \mathbb{Z}_{[1,n_\mathsf{J}]}$ such that $\mathsf{J}_h \subseteq \mathsf{S}_k$. Since $x(k)\in \hat{\mathcal{Z}}_k^i$ for every $i\in \mathsf{J}_h \subseteq \mathsf{S}_k$, the intersection in \eqref{eq:agreement_protocol} is non-empty and $x(k)\in \mathcal{I}_k^h$.
\end{proof}

We have shown that the measurement update $\hat{\mathcal{Z}}_k^i$ of sensor~$i$ contains the true state $x(k)$ if it is uncompromised by the attacker at time~$k$. Thus, the intersection between the measurement updates of any subset of uncompromised sensors must also contain $x(k)$. In the proof of Theorem~\ref{thm:safe-agreement}, we show that by removing the number of attacked sensors $q$, we guarantee the existence of at least one subset $\mathsf{J}_h$ with cardinality $|\mathsf{J}_h|=c_\mathsf{J}$ which is attack-free. Therefore, the corresponding agreement set $\mathcal{I}_k^h$ must be non-empty and contain $x(k)$.

It is equivalently true to say that if the agreement set $\mathcal{I}_k^h$ is empty, then there must be at least one sensor in $\mathsf{J}_h$ that the attacker compromises. Notice that this is only a sufficient condition for detecting an attack because a stealthy attacker can design an attack that ensures that $\mathcal{I}_k^h$ is non-empty at any time~$k$. However, this comes at a cost to the attacker, which will be discussed later.

\subsection{Estimated Set}
\label{subsec_estimatedSet}

The estimated set is obtained by taking the union of all agreement sets:
\begin{equation}
\label{eq:estimated_set}
    \hat{\mathcal{X}}_k := \bigcup_{h\in \mathbb{Z}_{[1, n_\mathsf{J}]}} \mathcal{I}_k^h, \quad k\in \mathbb{Z}_{\geq 1}
\end{equation}
where $\hat{\mathcal{X}}_0 = \mathcal{X}_0$. 

\begin{thm} 
\label{thm:estimated-inclusion}
Let Assumption~\ref{assum:main} hold. Then, for every time~$k\in\mathbb{Z}_{\geq 1}$, the inclusion $x(k)\in\hat{\mathcal{X}}_k$ is guaranteed, where the estimate $\hat{\mathcal{X}}_k$ is computed according to \eqref{eq:estimated_set}.
\end{thm}
\begin{proof}
We prove this result by induction. By Assumption~\ref{assum:main}(\ref{assum:initial_set}), we have $x(0)\in\mathcal{X}_0=\hat{\mathcal{X}}_0$. Because $w(0)\in\mathcal{W}$, the inclusion holds for the time update at $k=1$,
$
x(1)\in\hat{\mathcal{X}}_{1|0}=A\mathcal{X}\oplus B u(0) \oplus\mathcal{W}.
$
Therefore, Theorem~\ref{thm:safe-agreement} can be applied, implying the existence of $h\in\mathbb{Z}_{[1,n_\mathsf{J}]}$ such that $x(1)\in\mathcal{I}_1^h$, which gives
$
x(1)\in\hat{\mathcal{X}}_1 = \bigcup_{h\in\mathbb{Z}_{[1,n_\mathsf{J}]}} \mathcal{I}_1^h.
$
This, in turn, implies that the time update at $k=2$ contains the true state,
$
x(2)\in\hat{\mathcal{X}}_{2|1} = A\hat{\mathcal{X}}_1 \oplus Bu(1) \oplus \mathcal{W}.
$

Now, let us assume $x(k')\in\hat{\mathcal{X}}_{k'|k'-1}$ for some $k'\in\mathbb{Z}_{\geq 2}$. Then, by Theorem~\ref{thm:safe-agreement}, there exists $h\in\mathbb{Z}_{[1,n_c]}$ such that $x(k')\in\mathcal{I}_{k'}^h$, implying
$
x(k')\in\hat{\mathcal{X}}_{k'} = \bigcup_{h\in\mathbb{Z}_{[1,n_\mathsf{J}]}} \mathcal{I}_{k'}^h.
$
Therefore, we have the inclusion at the next time update,
$
x(k'+1)\in\hat{\mathcal{X}}_{k'+1|k'} = A\hat{\mathcal{X}}_{k'} \oplus Bu(k') \oplus \mathcal{W}.
$
Thus, the proof is completed because we showed that, for every $k\in\mathbb{Z}_{\geq 1}$, $x(k-1)\in\hat{\mathcal{X}}_{k-1|k-2}$ implies $x(k-1)\in\hat{\mathcal{X}}_{k-1}$, which, in turn, implies $x(k)\in\hat{\mathcal{X}}_{k|k-1}$. Hence, $x(k)\in\hat{\mathcal{X}}_k$.
\end{proof}

Although the above theorem guarantees the inclusion of the true state, it is important to remark that the number of (constrained) zonotopes in the estimated set $\hat{\mathcal{X}}_k$ may increase with respect to time if the attack is stealthy. We address this issue in Section~\ref{subsec_complexity} by proposing several techniques that manage the computational efficiency of the algorithm.
However, the proposed estimation algorithm is resilient because the attacker cannot arbitrarily degrade the estimation accuracy over time. If a sensor is injected with a large-valued attack signal, it will be automatically discarded by either yielding an empty measurement update or an empty agreement set (see Section~\ref{sec:attack-detection-sec}). 

\subsection{Illustrative and numerical examples} \label{sec:examples_num}

Fig.~\ref{fig:cartoon_algo_1} illustrates Algorithm~\ref{alg:s3e} on a simple example of $p=3$ sensors, where $q=1$ sensor has been compromised. 
For illustration purposes, we choose all sensors to have the same dimension $\mathbb{R}^{n_y}$ and the first sensor to be compromised. 
For each sensor~$i$, we obtain its corresponding output measurement set $\mathcal{Y}_{k}^{i}$ in the output space 
as depicted in Fig.~\ref{fig:cartoon_algo_1}, which do not necessarily intersect as each sensor~$i$ can be measuring different components of the state $x(k)\in\mathbb{R}^{n_x}$. 

Next, the time update set $\hat{\mathcal{X}}_{k|k-1}$ is constructed according to \eqref{eq:time_update} and shown in the second box from the top in Fig.~\ref{fig:cartoon_algo_1}. Each output measurement set $\mathcal{Y}_{k}^{i}$ is mapped back to the state space $\mathbb{R}^{n_x}$ and constrained to be in the time update set $\hat{\mathcal{X}}_{k|k-1}$ as shown in the third box from the top in Fig.~\ref{fig:cartoon_algo_1}. The intersection of the measurement update sets $\hat{\mathcal{Z}}^{i}_{k}$ is informative as the intersections of the attack-free $\hat{\mathcal{Z}}^{i}_{k}$ is non-empty. However, the measurement update set $\hat{\mathcal{Z}}^{i}_{k}$ corresponding to a sensor~$i$ that is corrupted could also intersect with the other attack-free measurement update sets. This scenario is depicted in Fig.~\ref{fig:cartoon_algo_1} where the measurement update set $\hat{\mathcal{Z}}^{1}_{k}$ of the compromised sensor $1$ intersects with the attack-free measurement update set $\hat{\mathcal{Z}}^{2}_{k}$ due to an intelligently designed attack signal $a_1$. We discuss classes of attack signals that can be detected by our algorithm in Section~\ref{sec:attack-detection-sec}. 

Since the identity of the compromised sensor(s) is unknown at any given time step $k$, but the number of compromised sensors $q$ is known, we have to check for non-empty intersections of all $p-q$ combinations of the measurement update sets $\hat{\mathcal{Z}}_{k}^{i}$, which totals to $\binom{p}{p-q}$ checks. Its union then forms our estimated set $\hat{\mathcal{X}}_{k}$ as depicted in the bottom box of Fig.~\ref{fig:cartoon_algo_1}.

\begin{figure}[!th]
    \centering
    \includegraphics[width=0.475\textwidth]{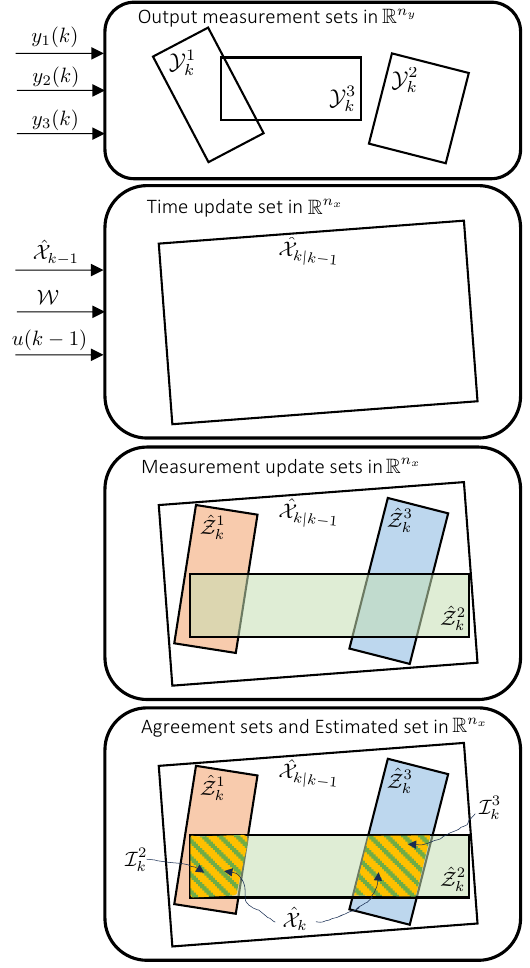}
    \caption{Illustration of Algorithm 1 with $p=3$ and $q=1$. The algorithm deduces that either sensor $1$ or $3$ has been compromised since their corresponding measurement update sets $\hat{\mathcal{Z}}^{1}_{k}$ and $\hat{\mathcal{Z}}^{3}_{k}$ do not intersect.}
    \label{fig:cartoon_algo_1}
\end{figure}

To further elucidate Algorithm \ref{alg:s3e}, we consider a simple numerical example. Simulations for more realistic systems are provided later in Section \ref{sec:evaluation}. 
\begin{ex} \label{example:running}
Consider the following system
\begin{eqnarray*} 
    x(k+1) & = & 0.1 x(k) + w(k), \nonumber \\ 
    y_i(k) & = &x(k)+v_i(k)+a_i(k), \qquad i\in\{1,2,3\}
\end{eqnarray*}
where for all $k\in\mathbb{Z}_{\geq 0}$, the  disturbance and noise are $w(k)\in[-10,10]$ and $v_i(k)\in[-0.1,0.1]$, respectively. Notice that this system is observable via each sensor~$i$. For the purposes of this illustration, suppose $x(0)=0$, $\hat{\mathcal{X}}_0 = [-100,100]$ and the disturbance and noise at time step $k=1$ are $w(1)=4$, $v_1(1)=0.03$, $v_2(1)=-0.02$, $v_3(1)=-0.07$, respectively. Notice also that $x(1)=4$.

Suppose sensor $1$ and $2$ are under attack and sensor $3$ is uncompromised, i.e.,  $p=3$, $q=2$, $c_{\mathsf{J}}=1$, the parameter of Algorithm \ref{alg:s3e} is $n_{\mathsf{J}}=3$, and $a_3(k)=0$ for all $k\in\mathbb{Z}_{\geq 0}$. Further, suppose the attacks on sensors $1$ and $2$ at time step $k=1$ are $a_1(1)=-10.03$ and $a_2(1)=-10$. We construct the following sets according to Algorithm \ref{alg:s3e} for time step $k=1$: 
\begin{itemize}
    \item The measurement output set for each sensor according to \eqref{eq:output_sets} is $\mathcal{Y}^{1}_{k}=[-6.1,-5.9]$, $\mathcal{Y}^{2}_{k}=[-6.12,-5.92]$ and $\mathcal{Y}^{3}_{k}=[3.83,4.03]$, respectively.
    \item The time update set according to \eqref{eq:time_update}: $\hat{\mathcal{X}}_{1|0}= 0.1 \hat{\mathcal{X}}_0 \oplus \mathcal{W} = [-10,10] \oplus [-10,10] = [-20,20]$. 
    \item Measurement update sets according to \eqref{eq:meas_update_i}:  $\hat{\mathcal{Z}}^{1}_{1} = [-6.1,-5.9]$, $\hat{\mathcal{Z}}^{2}_{1} = [-6.12,-5.92]$, $\hat{\mathcal{Z}}^{3}_{1}=[3.83,4.03]$.
    \item Agreement sets using \eqref{eq:agreement_protocol}: for $\mathsf{J}_1=\{1\}$, $\mathsf{J}_2=\{2\}$ and $\mathsf{J}_3=\{3\}$, we get $\mathcal{I}^{1}_{1}= \hat{\mathcal{Z}}_{1} = [-6.1,-5.9]$, $\mathcal{I}^{2}_{1}= \hat{\mathcal{Z}}^{2}_1 = [-6.12,-5.92]$, and $\mathcal{I}^{3}_{1}= \hat{\mathcal{Z}}^{3}_1 = [3.83,4.03]$. 
    \item Finally, the estimated set using \eqref{eq:estimated_set} is $\hat{\mathcal{X}}_{1}=\mathcal{I}^{1}_{1} \cup \mathcal{I}^{2}_{1} \cup \mathcal{I}^{3}_{1} = [-6.12, -5.9] \cup [3.83, 4.03]$.
\end{itemize}
\end{ex}

By Theorem \ref{thm:estimated-inclusion}, the true state $x(1)$ at time step $k=1$ can lie either in the interval $[-6.12,-5.9]$ or $[3.83,4.03]$. The disparity between the two estimated intervals highlights the limitation of the S3E algorithm: is the true state contained in one interval or the other?  Nonetheless, having multiple sets of reasonable sizes that can potentially contain the true state provides utility in safety-critical applications, where the estimated set verifies whether the system's trajectory could have entered regions of the state space designated as unsafe.  

\subsection{S3E Algorithm vs. Point-based Methods}
\label{sec:set-vs-point-based}
Point-based secure state estimators provide a point-based estimate $\hat{x}(k)\in \mathbb{R}^{n_x}$. The security guarantee for these estimators against sensor attacks is an estimation error bound that only depends on the noise and is independent of the attack signal; see \cite{chong2015, pajic2016, shoukry2017, kim2018, mitra2019, chong2020, lee2020, he2021}. However, for providing such a security guarantee, the fundamental assumption is that the number of attacked sensors is strictly less than half of the total number of sensors. The estimation error bound for point-based estimators has the following form
\begin{equation}
\label{eq:errorbound-pointbased}
    \|\hat{x}(k) - x(k) \| \leq \beta(\|\hat{x}(0)-x(0)\|, k) + \gamma_1(\|w\|_{\ell^\infty})  + \gamma_2(\|v\|_{\ell^\infty}) + \gamma_3(\|a\|_{\ell^\infty}) \mathbbm{1}_{q\geq \frac{p}{2}}
\end{equation}
where $\beta$ is a $\mathcal{KL}$ function, $\gamma_1,\gamma_2,\gamma_3$ are $\mathcal{K}_\infty$ functions\footnote{A continuous function $\gamma:\mathbb{R}_{\geq 0}\to\mathbb{R}_{\geq 0}$ is a class $\mathcal{K}$ function, if it is strictly increasing and $\gamma(0)=0$. It is a class $\mathcal{K}_\infty$ function if it is also unbounded. A continuous function $\beta:\mathbb{R}_{\geq0}\times \mathbb{R}_{\geq 0} \to \mathbb{R}_{\geq 0}$ is a class $\mathcal{KL}$ function, if: (i) $\beta(.,s)$ is a class $\mathcal{K}$ function for any $s\geq 0$; (ii) $\beta(r,.)$ is non-increasing and (iii) $\beta(r,s)\to 0$ as $s\to \infty$ for any $r\geq 0$.}, and $\mathbbm{1}_{q\geq \frac{p}{2}}$ is the indicator function, which is $1$ when $q\geq \frac{p}{2}$ and $0$ otherwise. Here, $v=[\arraycolsep=2pt\begin{array}{ccc} v_1^\top & \dots & v_p^\top \end{array}]^\top$ and $a=[\arraycolsep=2pt\begin{array}{ccc} a_1^\top & \dots & a_p^\top \end{array}]^\top$ concatenate the measurement noise and attack signals in a column vector, respectively.

It is easy to see that if the number of sensors $p$ is not strictly twice the number of compromised sensors $q$ (i.e., the assumption of $q<p/2$ is violated), the estimation error bound depends on the magnitude of the attack signal. In this case, the attacker can arbitrarily degrade the accuracy of point-based estimators by injecting attack signals of arbitrarily large magnitude. On the other hand, our set-based approach avoids this situation because it restricts the attack space by allowing only small-magnitude signals to go undetected. Large attack signals are automatically detected and discarded by either the measurement update \eqref{eq:meas_update_i} or the agreement protocol \eqref{eq:agreement_protocol}.

Another important advantage of our set-based approach is that the estimated set $\hat{\mathcal{X}}_k$ contains a collection of agreement sets, which serve as multiple hypotheses for the true state $x(k)$ in the state space\footnote{Although beyond the scope of this paper, an interesting research direction is developing a hypothesis testing technique to figure out the true set from the sets influenced by the attack signals.}.  
Point-based estimators, on the other hand, have an estimation error bound \eqref{eq:errorbound-pointbased} that indicates a single ball inside which the true state might lie. In the presence of attacks on more than half the number of sensors, the size of this ball depends on the magnitude of the attacks. Thus, as described in Example~\ref{example:running}, our set-based method can at least identify plausible regions in the state space containing the true state. Moreover, even when less than half the number of sensors are compromised, bounds like \eqref{eq:errorbound-pointbased} are very conservative in practice because the comparison functions $\beta\in\mathcal{KL}$ and $\gamma\in\mathcal{K}_\infty$ are hard to estimate. 

On the other hand, our set-based SSE guarantees that the true state $x(k)$ of system~\eqref{eq:system} lies in at least one of the agreement sets in the estimated collection of sets $\hat{\mathcal{X}}_k$ at each $k\in\mathbb{Z}_{\geq 0}$ (see  Theorem~\ref{thm:estimated-inclusion}), it must also lie in a zonotope that overbounds $\hat{\mathcal{X}}_k$. That is, let $\hat{\mathcal{X}}_k$ be overbounded by a zonotope 
\begin{equation}
\label{eq:overbounding_zono}
\bar{\mathcal{X}}_k=\zono{\hat{c}_x(k),\hat{G}_x(k)}
\end{equation}
which is obtained by solving
\begin{equation} \label{eq:zono_inclusion}
\min \text{rad}(\bar{\mathcal{X}}_k) ~\text{subject to}~ \hat{\mathcal{X}}_k \subseteq \bar{\mathcal{X}}_k.
\end{equation}
Then, considering $\hat{c}_x(k)$ to be an estimate of $x(k)$, the estimation error can be bounded by
\[
\|\hat{c}_x(k) - x(k)\| \leq \text{rad}(\bar{\mathcal{X}}_k).
\]
We remark that this error bound is significantly smaller in practice than the error bounds obtained by point-based secure estimators.
To illustrate this point, we revisit Example~\ref{example:running} for the case where fewer than half the number of sensors are attacked.

\begin{ex} \label{ex:compare}
    Reconsider Example~\ref{example:running}, but suppose sensor~$1$ is under attack while sensors~$2$ and $3$ are uncompromised. We have the number of sensors $p=3$ and the number of attacked sensors $q=1$, and, for the sake of this illustration, assume the attack signals $a_2(k)=a_3(k)=0$, for all $k\in\mathbb{Z}_{\geq 0}$. Furthermore, suppose the attack on sensor~$1$ is $a_1(1)=-10.03$. 
    We obtain the estimated set according to Algorithm~\ref{alg:s3e} at time step $k=1$ as follows: 
    \begin{itemize}
    \item The measurement output set for each sensor is $\mathcal{Y}^{1}_{k}=[-6.1,-5.9]$ and $\mathcal{Y}^{3}_{k}=[3.83,4.03]$ as calculated before, and $\mathcal{Y}^{2}_{k}=[3.88,4.08]$ according to \eqref{eq:output_sets}.
    \item The time update set $\hat{\mathcal{X}}_{1|0}= [-20,20]$.
    \item Measurement update sets according to \eqref{eq:meas_update_i}:  $\hat{\mathcal{Z}}^{1}_{1}= [-6.1,-5.9]$, $\hat{\mathcal{Z}}^{2}_{1} = [3.88,4.08]$, and 
    $\hat{\mathcal{Z}}^{3}_{1}=[3.83,4.03]$.
    \item Agreement sets using \eqref{eq:agreement_protocol}: for $\mathsf{J}_1=\{1,2\}$, $\mathsf{J}_2=\{1,3\}$ and $\mathsf{J}_3=\{2,3\}$, we get $\mathcal{I}^{1}_{1}= \hat{\mathcal{Z}}^{1}_{1} \cap \hat{\mathcal{Z}}^{2}_{1} = \emptyset$, $\mathcal{I}^{2}_{1}= \hat{\mathcal{Z}}^{1}_{1} \cap \hat{\mathcal{Z}}^{3}_{1} = \emptyset$, and $\mathcal{I}^{3}_{1}= \hat{\mathcal{Z}}^{2}_{1} \cap \hat{\mathcal{Z}}^{3}_{1} = [3.88,4.03]$.
    \item Finally, the estimated set using \eqref{eq:estimated_set} is $\hat{\mathcal{X}}_{1}=\mathcal{I}^{1}_{1} \cup \mathcal{I}^{2}_{1} \cup \mathcal{I}^{3}_{1} = [3.88, 4.03]$. 
    \end{itemize}
    
    We consider the center $3.955$ of $\hat{\mathcal X}_1$ as an estimate of $x(1)=4$, which gives the estimation error $0.045$. 
    \end{ex}
    
    Under the conditions of Example \ref{ex:compare}, we compare our set-based SSE with two comparable point-based SSE in the literature, namely \cite{chong2015} and \cite{shoukry2018}. The comparison metrics we will employ are (i) estimation accuracy, and (ii) computational complexity. We do not compare attacked sensor detection rates, as they were not addressed by both existing point-based SSEs from \cite{chong2015} and \cite{shoukry2018}.

\subsubsection{Estimation accuracy}
For both of the point-based SSEs, we design each observer in a multi-modal Luenberger observer, for $\mathsf{J}\subset \{1,\dots,N\}$ with $|\mathsf{J}|=p-2q=1$:
\begin{align} \label{eq:point-based-obs-algo-example}
    \hat{x}_{\mathsf{J}}(k+1) & = 0.1 \hat{x}_{J}(k) + L_{\mathsf{J}} \left( \mathsf{1}_{|\mathsf{J}|} \otimes \hat{x}_{\mathsf{J}}(k)  - y_{\mathsf{J}}(k) \right),
\end{align}
with Luenberger gain $L_1 = L_2 = L_3 = -1$. Based on Theorem 6.2 of Shoukry et. al., the estimation error at $k=1$ is $4 \bar{\psi}^2 + 2 \bar{\mu}^2$, where $\bar{\psi}>0$ and $\bar{\mu}>0$  are the bounds on the process disturbance $w$ and measurement noise $v_i$ respectively, i.e., $\|w(k)\|\leq \bar{\psi}$ and $\|v_i(k)\| \leq \bar{\mu}$ for all $k\in\mathbb{Z}$ and $i\in\{1,2,3\}$. In Example \ref{ex:compare}, $\bar{\psi}=20$ and $\bar{\mu}=0.2$. We summarise the computed estimation error for our set-based SSE and the two point-based SSEs in \cite{chong2015} and \cite{shoukry2017} in Table \ref{tab:compare}.

\subsubsection{Computational requirements}
The main computational bottleneck across all the SSE algorithms is the amount of memory required at each sample time. This pertains directly to the length of the vector of variables involved in the algorithm's computation. The memory requirement for our set-based SSE is computed for the worst-case scenario when all the agreement sets $\mathcal{I}^{h}_{k}$ given \eqref{eq:agreement_protocol} are non-empty. The memory computation for the point-based algorithms is exact. Recall that $n_x$ is the dimension of the state vector $x$, and $\xi^h$ and $m^h$ are the number of generators and constraints for each agreement set $\mathcal{I}^{h}_{k}$ at time $k$ in \eqref{eq:agreement_protocol}. Note that the number of generators $\xi^h$ and constraints $m^h$ at each time $k$ can be reduced by a reduction algorithm. The resulting memory requirements for our set-based SSE, compared with point-based SSEs, are shown in Table \ref{tab:compare}.

\begin{table}[h!] 
    \centering
    \begin{tabular}{|p{0.17\linewidth}|p{0.33\linewidth}|p{0.23\linewidth}|p{0.17\linewidth}|}
        \hline 
        & Our Set-based SSE & Point-based SSE from [12] & Point-based SSE using SMT from \cite{shoukry2018}\\ \hline \hline
       Estimation accuracy $\|\hat{x}(1)-x(1)\|$ & $0.045$ & $0.836$ & $1600.08$ \\
        \hline
        Memory requirements & $n_x \binom{p}{p-q}+ \sum_{h=1}^{\binom{p}{p-q}} n_x \xi^h + m^h(\xi^h+1)$        & $n_x  \binom{p}{p-q} + n_x  \binom{p}{p-2q}
              $ & $n_x  \binom{p}{p-q}$ \\
        \hline
    \end{tabular}
    \caption{Comparison of the estimation and memory requirements of our set-based SSE versus comparable point-based SSEs.}
    \label{tab:compare}
\end{table}

\subsection{Strategies to Handle Computational Complexity}
\label{subsec_complexity}
A major drawback of Algorithm~\ref{alg:s3e} is that the attacker can design intelligent and stealthy attacks such that the number of agreement sets in the estimate $\hat{\mathcal{X}}_k$ keeps growing with time. Although generating such attacks is a difficult problem, as we argue in Section~\ref{sec:attack-detection-sec}, the computational complexity of the algorithm may increase exponentially with time~$k$ in the worst-case scenario. Therefore, it is important to reduce the complexity by adopting several pruning methods that are described below.

An obvious step is to remove any empty sets or subsets of other sets in the estimated set $\hat{\mathcal{X}}_k$. It is also possible to obtain a single overbounding zonotope of $\hat{\mathcal{X}}_k$ as in \eqref{eq:zono_inclusion}, and use it instead of $\hat{\mathcal{X}}_k$ in the next time update. However, a better tradeoff between accuracy and complexity is to not overbound the whole collection but only the intersecting zonotopes in the collection $\hat{\mathcal{X}}_k$. This may not make the cardinality of $\hat{\mathcal{X}}_k$ equal to one, but it reduces it significantly by allowing minimal loss of accuracy.
Another method is to employ zonotope reduction methods \cite{yang2018comparison} to reduce the number of generators in the zonotopes, which are often increased when performing the Minkowski sum operations.

\section{Attack Detection, Identification, and Filtering} \label{sec:attack-detection-sec}

In this section, we provide sufficient conditions that guarantee the security of the S3E algorithm.

\subsection{Attack Detection}

We say that the attack is \textit{detected} if the system operator (defender) knows for certain that at least one sensor is compromised at time~$k$, i.e., $|\mathsf{S}_k|<p$. While attack detection does not necessarily reveal the specific compromised sensor(s), it allows the system operator to raise the alarm and take necessary steps to secure the system. We provide a sufficient condition for attack detection below.

\begin{prop}
    \label{prop:detection-multisets}
    Consider the agreement sets $\mathcal{I}_k^h$, for $h\in \mathbb{Z}_{[1,n_\mathsf{J}]}$, given by \eqref{eq:agreement_protocol}. Then, if
    \begin{equation}
        \label{eq:attack-detection-suff-cond}
        \bigcap_{h\in \mathbb{Z}_{[1,n_\mathsf{J}]}} \mathcal{I}_k^h = \emptyset
    \end{equation}
    there is at least one sensor that is compromised at time~$k$.
\end{prop}
\begin{proof}
    We prove the proposition by showing its logical equivalence: `no sensor is compromised' implies `\eqref{eq:attack-detection-suff-cond} does not hold.' 
    Since for every $h\in\mathbb{Z}_{[1,n_\mathsf{J}]}$, the system is redundantly observable through every subset $\mathsf{J}_h$ of sensors with $|\mathsf{J}|=c_\mathsf{J}$.
    If no sensor is compromised, then, by Theorem~\ref{thm:safe-agreement}, we have that, for every $h\in\mathbb{Z}_{[1,n_\mathsf{J}]}$, the agreement set $\mathcal{I}_k^h$ contains the true state $x(k)$. Therefore, $\mathcal{I}_k^1 \cap \dots \cap \mathcal{I}_k^{n_\mathsf{J}} \neq \emptyset$.
\end{proof}

The above proposition shows that whenever the estimated set contains more than one disjoint agreement set, there must be at least one compromised sensor. 
Once the attack is detected, the system operator can take necessary measures to secure the system and bring it to safety.
If the attacker wants to stay stealthy, they must design attack signals that ensure that all the agreement sets yield a non-empty intersection. In the next subsections, we show that the magnitude of such a stealthy attack has to be sufficiently small.

\subsection{Attack Identification}

We say that the attack $a_i(k)$ is \textit{identified} if the system operator knows for certain that sensor~$i$ is compromised at time~$k$, i.e., $i\notin \mathsf{S}_k$. This attack identification allows us to automatically discard measurements from the compromised sensors in the S3E algorithm. We achieve attack identification through a measurement update~\eqref{eq:meas_update_i}, which involves a generalized intersection between the time update $\hat{\mathcal{X}}_{k|k-1}$ and the output measurement set $\mathcal{Y}_k^i$. 

If sensor~$i$ is attacked at time~$k$, we derive a sufficient condition on the magnitude of the attack signal $a_i(k)$ that would yield an empty measurement update set $\hat{\mathcal{Z}}_k^i$. This automatically identifies the attack on sensor~$i$ and discards its effect from the estimation algorithm.
In the next theorem, we use the notation $\bar{\mathcal{X}}_{k|k-1}$ for representing the overbounding zonotope of the time update $\hat{\mathcal{X}}_{k|k-1}$, i.e., $\hat{\mathcal{X}}_{k|k-1} \subseteq \bar{\mathcal{X}}_{k|k-1}$, where
\[
\bar{\mathcal{X}}_{k|k-1}\coloneq \zono{c_x(k|k-1), G_x(k|k-1)}.
\]

\begin{thm}
    \label{thm:attack-detect-meas-update}
    Suppose sensor~$i$ is attacked at time~$k$. Then, if
    \begin{equation}
    \label{eq:detection-meas-update}
         \|a_i(k)\| > \sqrt{m_i} [\rad(C_i\bar{\mathcal{X}}_{k|k-1}) + \rad(\mathcal{V}_i)] - \| C_i [x(k) - c_x(k|k-1)] + v_i(k) - c_{v_i}\|
    \end{equation}
    then the measurement update $\hat{\mathcal{Z}}_k^i = \emptyset$. Recall that $m_i$ is the dimension of the output $y_i$ of system \eqref{eq:system}.
\end{thm}
\begin{proof}
    Given a zonotope $C_i\bar{\mathcal{X}}_{k|k-1}$ with center $C_i c_x(k|k-1)$ and radius $\rad(C_i \bar{\mathcal{X}}_{k|k-1})$, we have that any point $p\in \mathbb{R}^{m_i}$ is outside the $C_i\bar{\mathcal{X}}_{k|k-1}$ if $\|C_i c_x(k|k-1) - p\| > \sqrt{m_i} \rad(C_i \bar{\mathcal{X}}_{k|k-1})$. Similarly, for $\mathcal{Y}_k^i$ whose center is $C_i x(k) + v_i(k) - c_{v_i} + a_i(k)$ and radius is $\rad(\mathcal{Y}_k^i)$, any $p\in \mathbb{R}^{m_i}$ is outside $\mathcal{Y}_k^i$ if $\|C_i x(k) + v_i(k) - c_{v_i} + a_i(k) - p\| > \sqrt{m_i} \rad(\mathcal{Y}_k^i)$. Here, the square root term $\sqrt{m_i}$ appears because the radius of a zonotope is a maximum norm $\|\cdot\|_\infty$, which relates to the Euclidean norm $\|\cdot\|$ as follows: for any vector $p\in \mathbb{R}^{m_i}$, $\|p\|\leq \sqrt{m_i} \|p\|_\infty$.

    If the distance between the centers of $C_i\bar{\mathcal{X}}_{k|k-1}$ and $\mathcal{Y}_k^i$ is more than $\sqrt{m_i}[\rad(\bar{\mathcal{X}}_{k|k-1}) + \rad(\mathcal{Y}_k^i)]$, then any point $x \in \bar{\mathcal{X}}_{k|k-1}$ will be such that $C_ix \notin \mathcal{Y}_k^i$. Therefore, in this case, the measurement update \eqref{eq:meas_update_i} is empty. In other words, we have $\hat{\mathcal{Z}}_k^i = \emptyset$ if
    \[
    \|C_i c_x(k|k-1) - [C_i x(k) + v_i(k) - c_{v_i} + a_i(k)]\| > \sqrt{m_i}[\rad(\bar{\mathcal{X}}_{k|k-1}) + \rad(\mathcal{Y}_k^i)].
    \]
    Moreover, using the triangle inequality, we obtain
    \[
    \|C_i[x(k) - c_x(k|k-1)] + v_i(k) - c_{v_i}\| + \|a_i(k)\| \geq
    \|C_i c_x(k|k-1) - [C_i x(k) + v_i(k) - c_{v_i} + a_i(k)]\|.
    \]
    Thus, \eqref{eq:detection-meas-update} is established by combining the above two inequalities and, from \eqref{eq:output_sets}, observing that $\rad(\mathcal{Y}_k^i) = \rad(\mathcal{V}_i)$.
\end{proof}

The sufficient condition in Theorem~\ref{thm:attack-detect-meas-update} provides a class of attack signals that can be identified at the measurement update step of the S3E algorithm. If the measurement update set $\hat{\mathcal{Z}}_k^i$ is empty, then it reveals that sensor~$i$ is compromised at time~$k$. 
However, the implication does not hold in the other direction. Nonetheless, while \eqref{eq:detection-meas-update} establishes a threshold above which an attack will lead to an empty generalized intersection, it does not guarantee that smaller attacks will be undetected. The overapproximations used in the proof of Theorem~\ref{thm:attack-detect-meas-update} can introduce conservatism, meaning that attacks slightly below the threshold might still be detectable.
In other words, the measurement update $\hat{\mathcal{Z}}_k^i$ could still be empty even if the attacker violates \eqref{eq:detection-meas-update}. This can happen if the attack signal is such that $C_i\hat{\mathcal{X}}_{k|k-1}$ and $\mathcal{Y}_k^i$ are inconsistent with each other at time~$k$, resulting in an empty generalized intersection. 

Notice that the bound \eqref{eq:detection-meas-update} depends on the specific realization of measurement noise $v_i(k)$ at time~$k$, which is unknown to both the attacker and the defender. One can further lower bound \eqref{eq:detection-meas-update} and obtain a deterministic bound as follows:
\[
\begin{array}{rcl}
    \|a_i(k)\| & > & \sqrt{m_i} [\rad(C_i\bar{\mathcal{X}}_{k|k-1}) + \rad(\mathcal{V}_i)] - \| C_i [x(k) - c_x(k|k-1)] + v_i(k) - c_{v_i}\| \\ [0.5em]
    &\geq & \sqrt{m_i} [\rad(C_i\bar{\mathcal{X}}_{k|k-1}) + \rad(\mathcal{V}_i)] - \| C_i [x(k) - c_x(k|k-1)]\| - \|v_i(k) - c_{v_i}\| \\ [0.5em]
    &\geq & \sqrt{m_i} [\rad(C_i\bar{\mathcal{X}}_{k|k-1}) + \rad(\mathcal{V}_i)] - \| C_i [x(k) - c_x(k|k-1)]\| - \sqrt{m_i} \rad(\mathcal{V}_i) \\ [0.5em]
    &=& \sqrt{m_i} \rad(C_i\bar{\mathcal{X}}_{k|k-1}) - \| C_i [x(k) - c_x(k|k-1)]\|.
\end{array}
\]
If the attacker wants to stay stealthy and effective, it can choose the attack signal $a_i(k)$, for every $i\in\mathbb{Z}_{[1,p]}$, satisfying
\[
\|a_i(k)\| \leq \sqrt{m_i} \rad(C_i\bar{\mathcal{X}}_{k|k-1}) - \| C_i [x(k) - c_x(k|k-1)]\|.
\]
Although this still does not ensure a non-empty $\hat{\mathcal Z}_k^i$, it may serve as a rough guideline for the attacker to generate a stealthy attack signal.

\subsection{Attack Filtering}
We say that the attack $a_i(k)$ is \textit{filtered out} if it is detected and discarded from the agreement protocol. To elucidate, let $i\in \mathsf{J}_h$, for some $h\in\mathbb{Z}_{[1,n_\mathsf{J}]}$. Then, if the attack signal $a_i(k)$ is filtered out, then the corresponding agreement set $\mathcal{I}_k^h$ is empty and does not affect the estimated set $\hat{\mathcal{X}}_k$. Although the attack is filtered out whenever it is identified, filtering does not imply identification in general.

In addition to the measurement update, the attacks could also be filtered out via the intersections in the agreement protocol \eqref{eq:agreement_protocol}, which could yield an empty agreement set $\mathcal{I}_k^h$ in case of inconsistencies between the measurement updates of the sensors in $\mathsf{J}_h$. However, unlike the measurement update, the attacks may not always be identified when some agreement sets are empty. This is because a compromised sensor~$i$ could be a part of multiple agreement sets $\mathcal{I}_k^h$, for $h\in \mathbb{Z}_{[1,n_\mathsf{J}]}$, some of which may turn out to be non-empty. One can identify the attack $a_i(k)$ if, for every $h\in\mathbb{Z}_{[1,n_\mathsf{J}]}$ such that $i\in \mathsf{J}_h$, we have $\mathcal{I}_k^h=\emptyset$.

To present the next proposition, we introduce some notation. Recall that $\mathsf{S}_k\subseteq \mathbb{Z}_{[1,p]}$ is the subset of uncompromised sensors at time~$k$. Similarly, let $\mathsf{A}_k = \mathbb{Z}_{[1,p]} \setminus \mathsf{S}_k$ be the subset of compromised sensors at time~$k$. Define
\[
\mathcal{I}_k^{h,\mathsf{A}_k} = \bigcap_{j\in \mathsf{J}_h \setminus \mathsf{S}_k} \hat{\mathcal{Z}}_k^j
\]
to be the agreement set of attacked sensors in $\mathsf{J}_h$ and let its overbounding zonotope be
\[
\bar{\mathcal{I}}_k^{h,\mathsf{A}_k} = \zono{c_{\mathcal{I}}^{h,\mathsf{A}_k}(k), G_{\mathcal{I}}^{h,\mathsf{A}_k}(k)}.
\]
Similarly, define
\[
\mathcal{I}_k^{h,\mathsf{S}_k} = \bigcap_{j\in \mathsf{J}_h \setminus \mathsf{A}_k} \hat{\mathcal{Z}}_k^j
\]
be the agreement set of uncompromised sensors in $\mathsf{J}_h$ and let its overbounding zonotope be
\[
\bar{\mathcal{I}}_k^{h,\mathsf{S}_k} = \zono{c_{\mathcal{I}}^{h,\mathsf{S}_k}(k), G_{\mathcal{I}}^{h,\mathsf{S}_k}(k)}.
\]

\begin{prop}
    \label{prop:detection-agreementset}
    Let $\mathsf{J}_h\subset\mathbb{Z}_{[1,p]}$ be a subset of sensors with $|\mathsf{J}_h|=c_\mathsf{J}$. Then, if
    \begin{equation}
        \label{eq:attack-detect-agreement-protocol}
        \| c_{\mathcal{I}}^{h,\mathsf{A}_k}(k) \| > \sqrt{n_x} [\rad(\bar{\mathcal{I}}_k^{h,\mathsf{S}_k}) + \rad(\bar{\mathcal{I}}_k^{h,\mathsf{A}_k})] - \|c_{\mathcal{I}}^{h,\mathsf{S}_k}(k)\|
    \end{equation}
    then the agreement set $\mathcal{I}_k^h = \mathcal{I}_k^{h,\mathsf{S}_k} \cap \mathcal{I}_k^{h,\mathsf{A}_k} = \emptyset$. Recall that $n_x$ is the dimension of the state $x(k)$ of system \eqref{eq:system}. 
\end{prop}
\begin{proof}
    The proof follows from a similar line of argument as the proof of Theorem~\ref{thm:attack-detect-meas-update}.
\end{proof}

The interpretation of the inequality in Proposition~\ref{prop:detection-agreementset} is not as straightforward as the one in Theorem~\ref{thm:attack-detect-meas-update}. If the attacker wants to be effective and remain unfiltered, they must devise an attack strategy that violates \eqref{eq:attack-detect-agreement-protocol} at every time instant. Otherwise, the attack is discarded due to an empty agreement set. To violate \eqref{eq:attack-detect-agreement-protocol}, the attacker must take into account the distance between the intersections of the measurement updates of the compromised and uncompromised sensors.

Notice that $c_{\mathcal{I}}^{h,\mathsf{A}_k}(k)$, which is the center of the overbounding zonotope $\bar{\mathcal{I}}_k^{h,\mathsf{A}_k}$, is actually the Chebyshev center of the agreement set $\mathcal{I}_k^{h,\mathsf{A}_k}$. This agreement set is formed by the intersections of the measurement update sets of the attacked sensors $\mathsf{A}_k$. The centers of these measurement update sets are influenced by the corresponding attack signals $a_j(k)$. Therefore, to achieve $\mathcal{I}_k^h\neq \emptyset$, it is necessary that the attacker coordinates the attacks on sensors in such a way that the condition of Proposition~\ref{prop:detection-agreementset} is violated. 

From the attacker's perspective, the condition of Proposition~\ref{prop:detection-agreementset} is necessary for $\mathcal{I}_k^h\neq \emptyset$, and is not sufficient. This points to several interesting open problems for designing effective attacks. Beyond Theorem~\ref{thm:attack-detect-meas-update} and Propositions~\ref{prop:detection-multisets} and \ref{prop:detection-agreementset}, our future work will investigate whether a sufficient condition could be derived for the attacker to ensure that both the measurement updates and the agreement sets are non-empty. Moreover, the complexity of generating such attacks is also one of the important questions that will be deferred for our future work.

\subsection{Algorithm for Attack Filtering and Identification} \label{sec:attack_detection}
A notable contribution of the set-based state estimation scheme in this paper over other secure schemes is Assumption~\ref{assum:main}(\ref{assum:main_numAttacks}), which allows the attacker to compromise not only up to $p-1$ sensors at each time instant but also different subsets of sensors at different times. To this end, we remark that we can detect only those compromised sensors that are injected with non-stealthy attack signals, whereby non-stealthiness means the attack signals that satisfy the conditions of Theorem~\ref{thm:attack-detect-meas-update} and Propositions~\ref{prop:detection-multisets} and \ref{prop:detection-agreementset}. The attack identification algorithm is fairly simple and can be summarized as in Algorithm~\ref{alg:detection}.

\begin{algorithm}[h]
\caption{Attack filtering and identification}
\label{alg:detection}
\begin{algorithmic}[1]
\Require Measurement updates $\hat{\mathcal{Z}}_k^i$ for $i\in\mathbb{Z}_{[1,p]}$, sensor subsets $\mathsf{J}_h$, for $h\in\mathbb{Z}_{[1,n_\mathsf{J}]}$, and their corresponding agreement sets $\mathcal{I}_k^h$.

\State Initialize $\hat{\mathsf{S}}_k = \emptyset$

\State $\hat{\mathsf{A}}_k \leftarrow \{i\in\mathbb{Z}_{[1,p]} : \hat{\mathcal{Z}}_k^i = \emptyset\}$

\For{$h=1,\dots,n_\mathsf{J}$}

\If{$\mathcal{I}_k^h\neq\emptyset$}

\State Estimated safe subset $\hat{\mathsf{S}}_k \leftarrow \hat{\mathsf{S}}_k \cup \mathsf{J}_h$

\EndIf

\EndFor

\State Identified attacked sensors $\hat{\mathsf{A}}_k \leftarrow \hat{\mathsf{A}}_k \cup \{\mathbb{Z}_{[1,p]}\setminus \hat{\mathsf{S}}_k\}$.
\end{algorithmic}
\end{algorithm}

At time $k$, Algorithm~\ref{alg:detection} first identifies some of the attacked sensors by checking whether the corresponding measurement update is empty. Then, for each $h\in\mathbb{Z}_{[1,n_\mathsf{J}]}$, if the agreement set $\mathcal{I}_k^h$ is non-empty, then the sensors with indices in $\mathsf{J}_h$ are either safe or compromised with a stealthy attack signal. Otherwise, the set $\mathsf{J}_h$ contains at least one attacked sensor. By checking all the combinations $\mathsf{J}_h$ and storing a `potentially' safe subset of sensors in $\hat{\mathsf{S}}_k$ at every iteration, a subset of attacked sensors $\hat{\mathsf{A}}_k$ are updated by adding those sensors that are not in $\hat{\mathsf{S}}_k$.

\begin{rem}
    The estimated safe subset $\hat{\mathsf{S}}_k$ contains the true safe subset $\mathsf{S}_k$ at time $k$, i.e., $\mathsf{S}_k\subseteq\hat{\mathsf{S}}_k$. If some sensors are injected with small-valued stealthy attack signals, Algorithm~\ref{alg:detection} may not identify those attacks and consider those sensors to be uncompromised. Therefore, the estimated set of identified attacked sensors $\hat{\mathsf{A}}_k$ is only a subset of the true set of attacked sensors $\mathsf{A}_k = \mathbb{Z}_{[1,p]} \setminus \mathsf{S}_k$.
\end{rem}

\begin{rem}
    In Algorithm~\ref{alg:detection}, we identify a subset of attacked sensors at every time $k$. However, in the time-invariant attack setting where the attacker does not change the subset of sensors to compromise at every time~$k$, Algorithm~\ref{alg:detection} can be adapted to detect and remove the attacked sensors over time cumulatively. This has potential applications in sensor fault detection and isolation, as faults can be considered as naive, time-invariant attacks.
\end{rem}

\section{Numerical Simulation} 
\label{sec:evaluation}
We evaluated our proposed algorithms using two examples. In the first illustrative example, we assume that an attacker can target half of the total number of sensors at every time instant. We illustrate the estimation algorithm and use a strategy to reduce the complexity with a minimal loss of estimation accuracy. In the second example of a three-story building structure, we assume redundant observability from every subset of two sensors. We use the CORA toolbox \cite{althoff2015} to generate our simulations.

\subsection{Two-Dimensional Linear System}

In this example, we consider a simple two-dimensional linear system for illustrative purposes. We consider a zero control input $u(k)\equiv 0$, and the system matrices are given by
\begin{align*}
    A = &\left[\begin{array}{cc}
        1 & 0 \\
        1 & 1
    \end{array}
    \right], \quad
        C_1 = \left[\begin{array}{cc} 1 & 0 \\ 0 & 1\end{array}\right], \quad  C_2 = \left[\begin{array}{cc} 1 & 1 \\ 1 &0 \end{array}\right], \quad
        C_3 = \left[\begin{array}{cc} 0 & 1 \\ 1 & 0 \end{array}\right], \quad  C_4 = \left[\begin{array}{cc} 1 & 2 \\ 2 & 1 \end{array}\right].
\end{align*}
The number of sensors is $p=4$, and we suppose that an attacker can target any combination of $q=2$ sensors at different time instants. In this case, since $q=p/2$, the point-based state estimators requiring $q<p/2$ cannot be employed. 

Let the process noise bound be $\mathcal{W}=\zono{0,\sigma_\mathcal{W} I_2}$ with $\sigma_\mathcal{W} =0.02$
and the measurement noise bounds of the four sensors be
$
    \mathcal{V}_{1} = \mathcal{V}_{2} = \mathcal{V}_{3} = \mathcal{V}_{4} = \zono{0,\sigma_\mathcal{V} I_2}
$ with $\sigma_\mathcal{V} =1$, where $I_2$ is the identity matrix of dimension $2\times 2$. The attack is generated according to the following
\begin{align}
a_i(k) \sim U(-\sigma_\mathcal{V} \phi(k),\sigma_\mathcal{V} \phi(k))
\label{eq:attgen}
\end{align}
where $U$ is a uniform distribution over an interval $(-\sigma_\mathcal{V} \phi(k),\sigma_\mathcal{V} \phi(k))$ and $\phi(k)$ is strictly increasing sequence with $\phi(1)=1$. The index of the two attacked sensors rotates among the available sensor indices.

The time update, measurement update, agreement, and estimated sets can be computed using Algorithm~\ref{alg:s3e}. Fig.~\ref{fig:illustrative4Sensors} illustrates the time update set $\hat{\mathcal{X}}_{k|k-1}$ (green), the measurement sets $\hat{\mathcal{Z}}_k^i$ (pink (unattacked) and red (attacked)), the estimated set $\hat{\mathcal{X}}_{k}$ (black), and the true state $x(k)$ (blue) for different sets of sensors attacked at different times. The attacker chooses a random combination of two sensors to compromise at each time instant. At time~$k=1$, we have in Fig.~\ref{fig:k1_atts2s3} sensors $2$ and $3$ under attack. Then, sensors $3$ and $4$ are attacked in Fig.~\ref{fig:k2_atts3s4} at time step $k=2$. Notice that the true state stays included in the estimated set $\hat{\mathcal{X}}_k$ for every $k$.

The generated attack values lead to an increase in the number of generated sets, as depicted in Fig.~\ref{fig:k3_atts1s4}. To manage this complexity, we employ a reduction technique that involves taking the union of intersecting estimated sets. The impact of this technique is illustrated in Fig.~\ref{fig:illustrative4SensorsReduce}, where we use the same seed for generating random attacks from \eqref{eq:attgen}. This approach is illustrated by comparing Fig.~\ref{fig:k1_atts2s3} with Fig.~\ref{fig:k1_atts2s3B}. The reduced estimated sets are then carried forward to subsequent steps, as shown in Fig.~\ref{fig:k2_atts3s4B} and Fig.~\ref{fig:k3_atts1s4B}. This is why the sets in these figures show slight differences when compared to those in Fig.~\ref{fig:k2_atts3s4} and Fig.~\ref{fig:k3_atts1s4}.

\begin{figure*}[!htbp]
    \centering
    \begin{subfigure}[t]{0.32\textwidth}
     \centering
        \includegraphics[width=\textwidth]{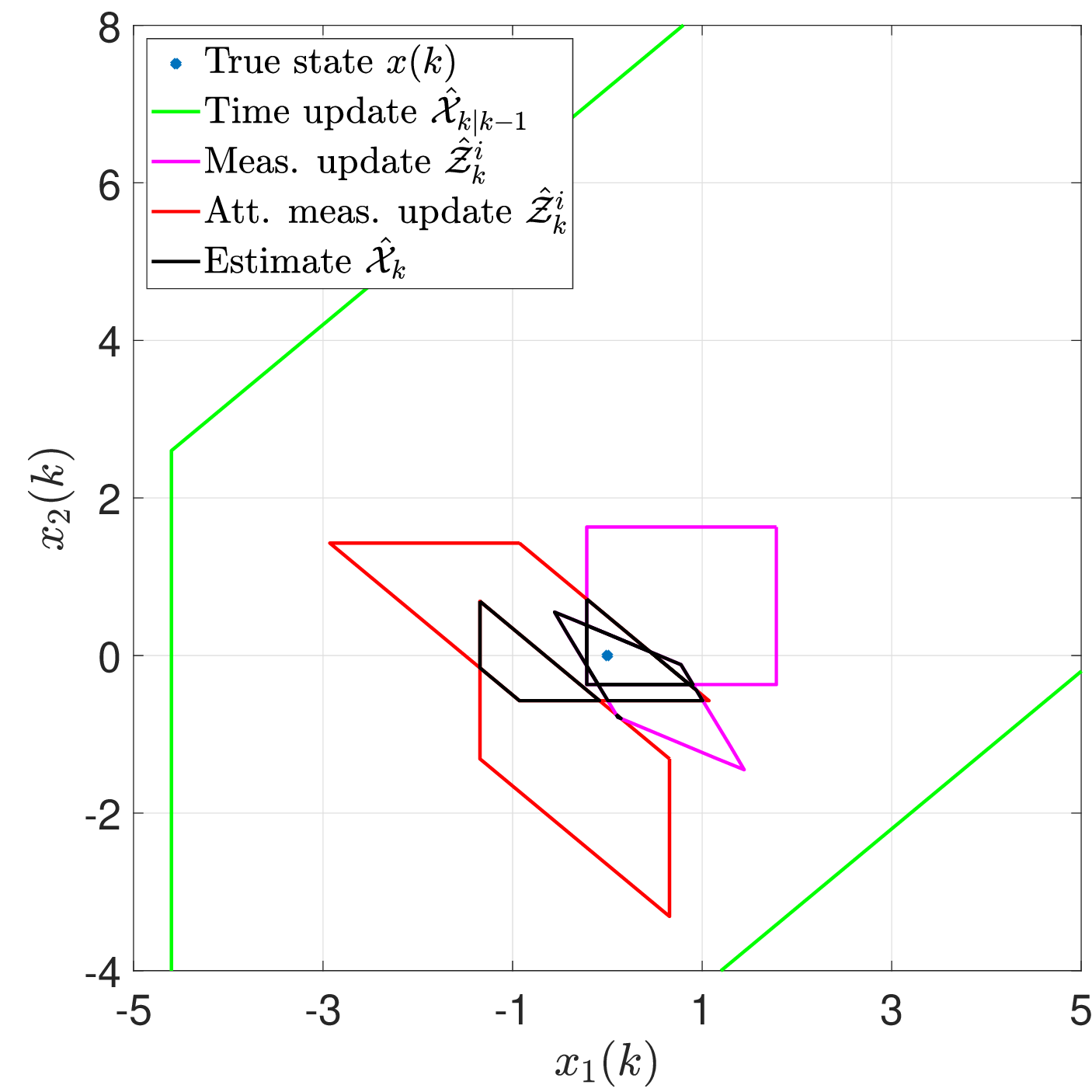}
        \caption{Time $k=1$. Sensors $\{2,3\}$ are attacked.}
        \label{fig:k1_atts2s3}
    \end{subfigure}
    \begin{subfigure}[t]{0.32\textwidth}
     \centering
        \includegraphics[width=\textwidth]{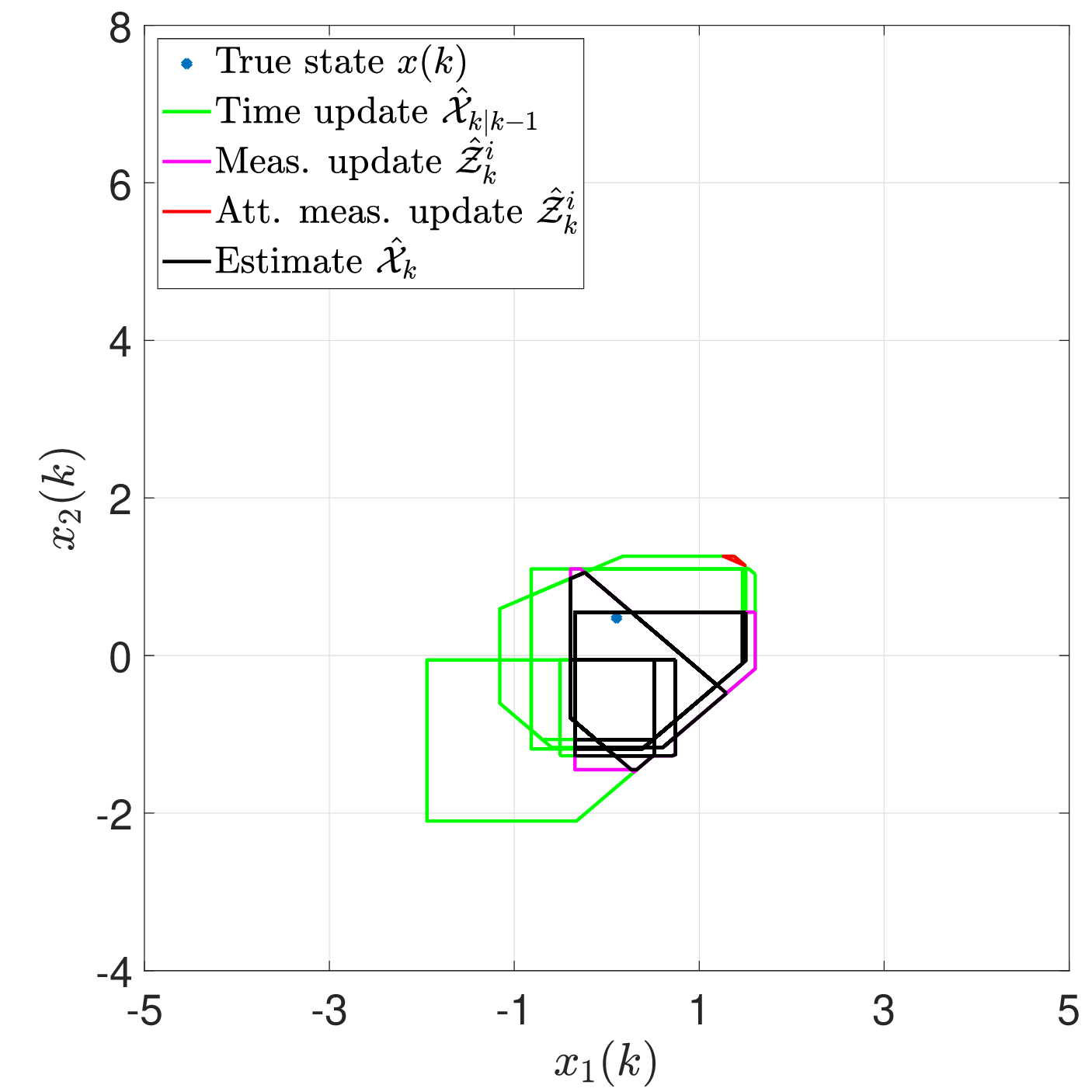}
        \caption{Time $k=2$. Sensors $\{3,4\}$ are attacked.}
        \label{fig:k2_atts3s4}
    \end{subfigure}
    \begin{subfigure}[t]{0.32\textwidth}
     \centering
        \includegraphics[width=\textwidth]{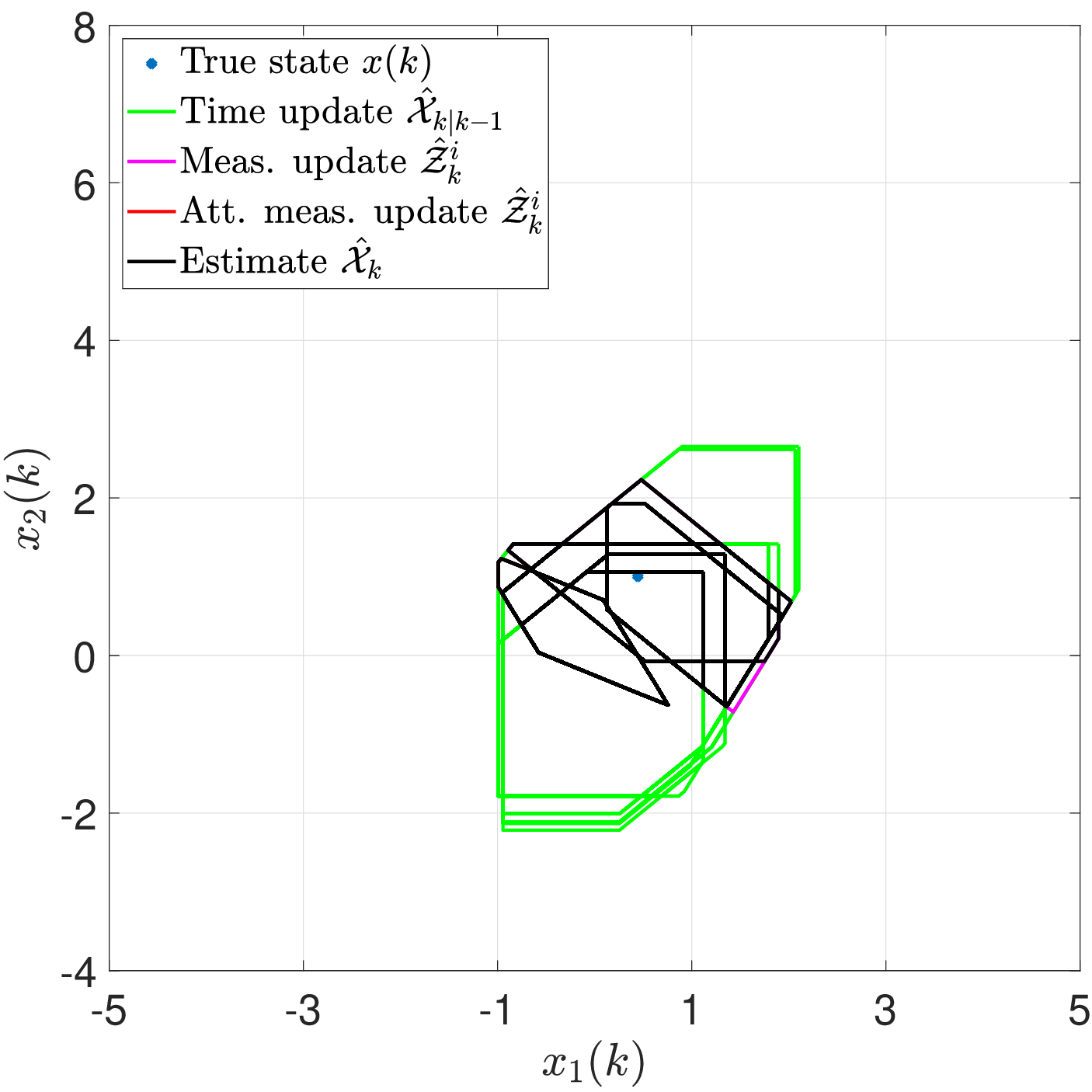}
        \caption{Time $k=3$. Sensors $\{1,4\}$ are attacked.}
        \label{fig:k3_atts1s4}
    \end{subfigure}
\caption{Snapshots of estimated sets using Algorithm~\ref{alg:s3e}.} 
    \label{fig:illustrative4Sensors}
\end{figure*}

\begin{figure*}[!htbp]
    \centering
    \begin{subfigure}[t]{0.32\textwidth}
     \centering
        \includegraphics[width=\textwidth]{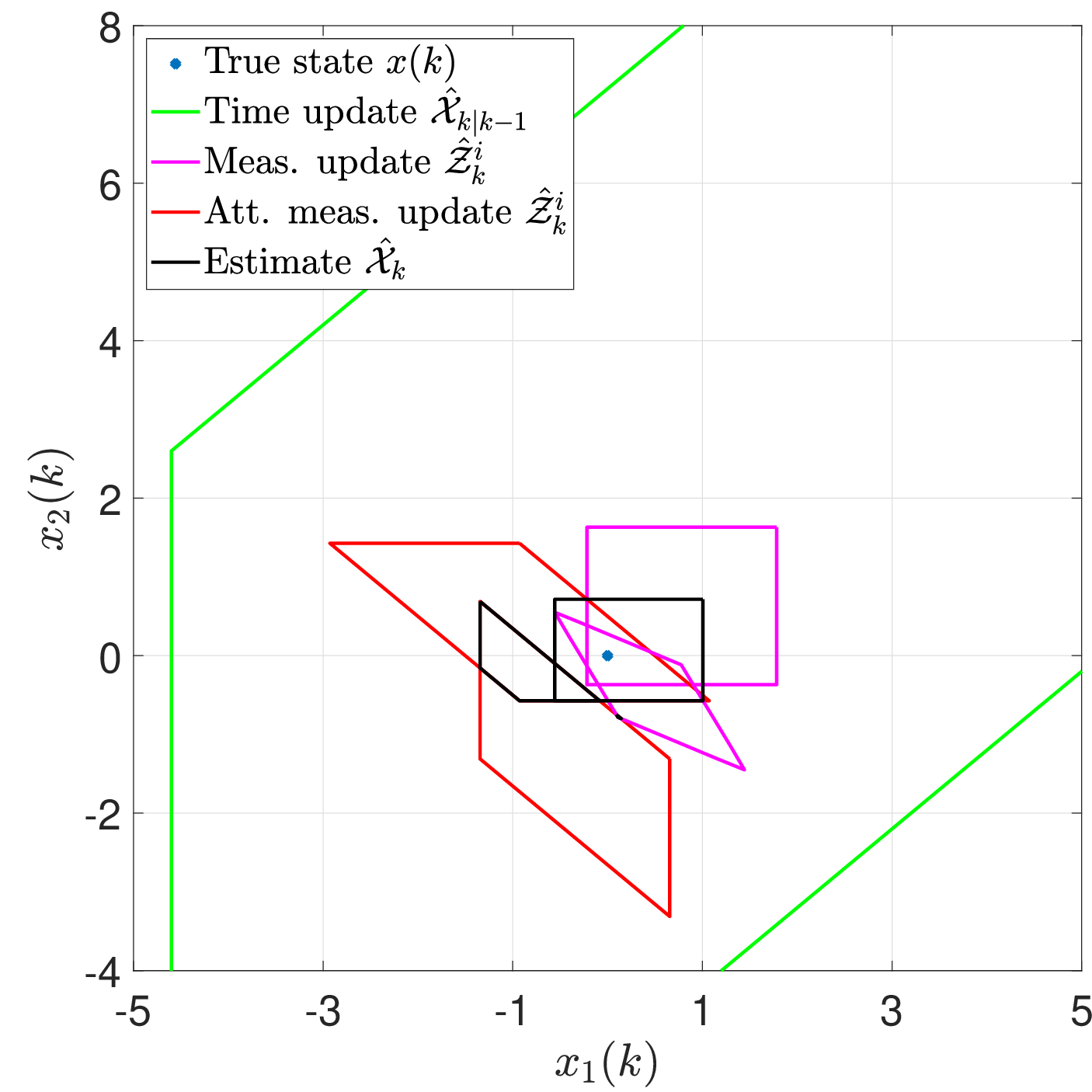}
        \caption{Time $k=1$. Sensors $\{2,3\}$ are attacked.}
        \label{fig:k1_atts2s3B}
    \end{subfigure}
    \begin{subfigure}[t]{0.32\textwidth}
     \centering
        \includegraphics[width=\textwidth]{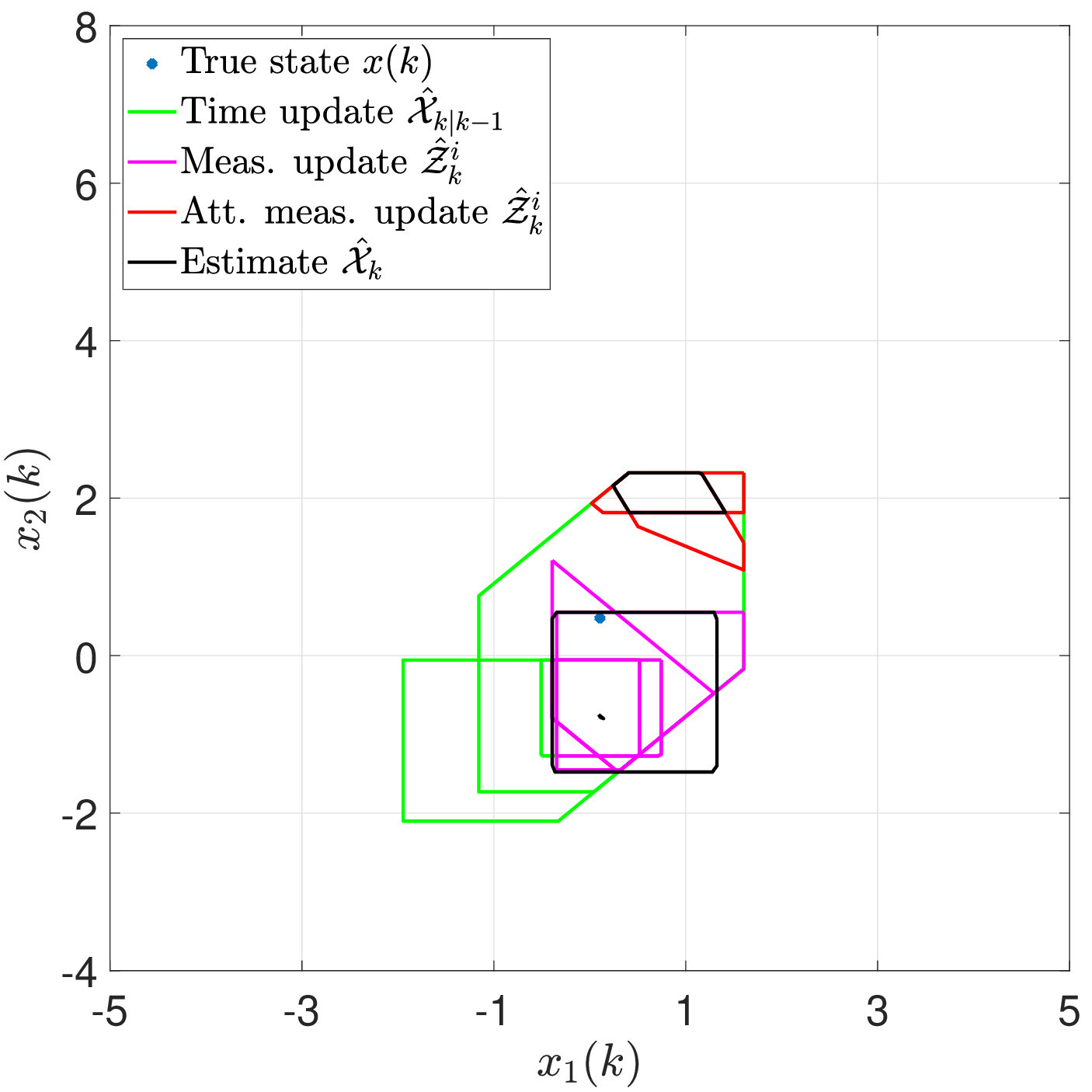}
        \caption{Time $k=2$. Sensors $\{3,4\}$ are attacked.}
        \label{fig:k2_atts3s4B}
    \end{subfigure}
    \begin{subfigure}[t]{0.32\textwidth}
     \centering
        \includegraphics[width=\textwidth]{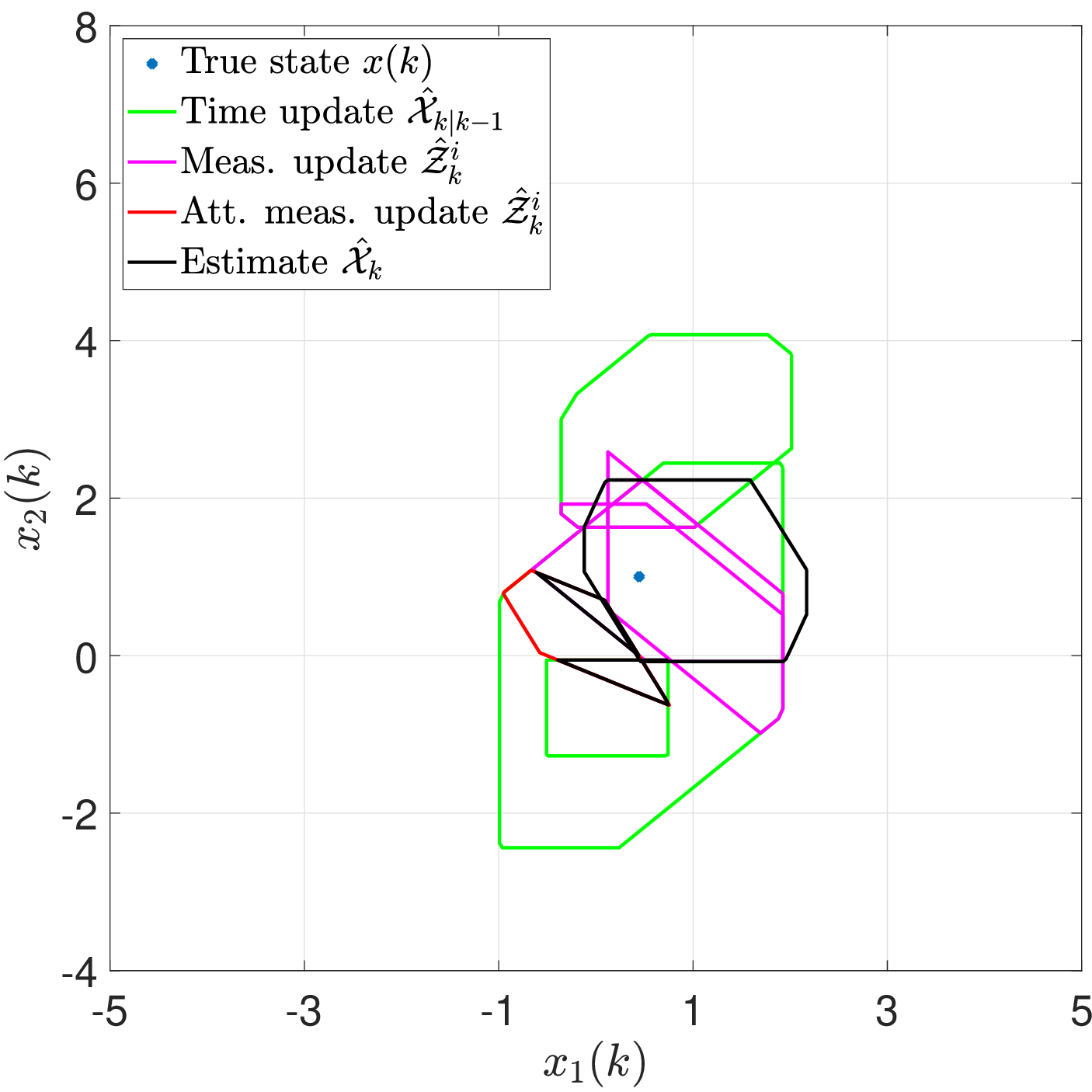}
        \caption{Time $k=3$. Sensors $\{1,4\}$ are attacked.}
        \label{fig:k3_atts1s4B}
    \end{subfigure}
\caption{Snapshots of estimated sets using Algorithm~\ref{alg:s3e} while applying a complexity reduction technique to overapproximate the estimated intersecting sets by one set.}
    \label{fig:illustrative4SensorsReduce}
\end{figure*}

\subsection{Three-story Building Structure}

We now consider a three-story building structure of \cite{truong2021} described by a mechanical system
\begin{equation}
\label{building_sys}
M\ddot{q}(t) + D\dot{q}(t) + Sq(t) = G u(t),
\end{equation}
where $q(t)\in\mathbb{R}^3$ is the vector of relative horizontal displacements of the floors and $u(t)\in\mathbb{R}$ is the ground acceleration due to an earthquake, which is a measured input signal. Also, $M\in\mathbb{R}^{3\times 3}$ is the mass matrix, $D\in\mathbb{R}^{3\times 3}$ is the damping matrix, $S\in\mathbb{R}^{3\times 3}$ is the stiffness matrix, and $G\in\mathbb{R}^3$ is the loading vector. 
The parameter values of the system \eqref{building_sys} are provided by \cite{truong2021} as:
\begin{align*}
    M &= \diag([\begin{array}{ccc} 478350 & 478350 & 517790 \end{array}]) \quad \text{(kg)} \\
    D &= 10^5 \times \left[\begin{array}{ccc}
        7.7626 & -3.7304 & 0.6514 \\
        -3.7304 & 5.8284 & -2.0266 \\
        0.6514 & -2.0266 & 2.4458
    \end{array}\right] \quad \text{(Ns/m)} \\
    S &= 10^8 \times \left[\begin{array}{ccc}
        4.3651 & -2.3730 & 0.4144 \\
        -2.3730 & 3.1347 & -1.2892 \\
        0.4144 & -1.2892 & 0.9358
    \end{array}\right] \quad \text{(N/m)} \\
    G &= [\begin{array}{ccc} 478350 & 478350 & 517790 \end{array}]^\top \quad \text{(kg)}.
\end{align*}
By considering the state $x(t)=[\begin{array}{cc} q(t)^\top & \dot{q}(t)^\top \end{array}]^\top$, we can obtain the state-space representation in continuous time
\[
\dot{x} = A_c x + B_c u
\]
where
\[
A_c = \left[\begin{array}{cc}
    0_{3\times 3} & I_3 \\
    -M^{-1}S & -M^{-1}D
\end{array}\right], \quad B_c = \left[\begin{array}{c}
    0_{3\times 1} \\ -M^{-1}G
\end{array}\right].
\]
After discretization with sample time $\delta=10^{-3}$, we obtain the system in the form \eqref{eq:system_state}, where
\begin{align*}
A = \exp(A_c \delta), \quad
B = A_c^{-1}(A - I_6) B_c.
\end{align*}
Here, our goal is to monitor the building dynamics under an earthquake using a secure set-based state estimation algorithm.

\begin{figure*}[!ht]
    \centering
    \begin{subfigure}[h]{0.32\textwidth}
     \centering
        \includegraphics[width=\textwidth]{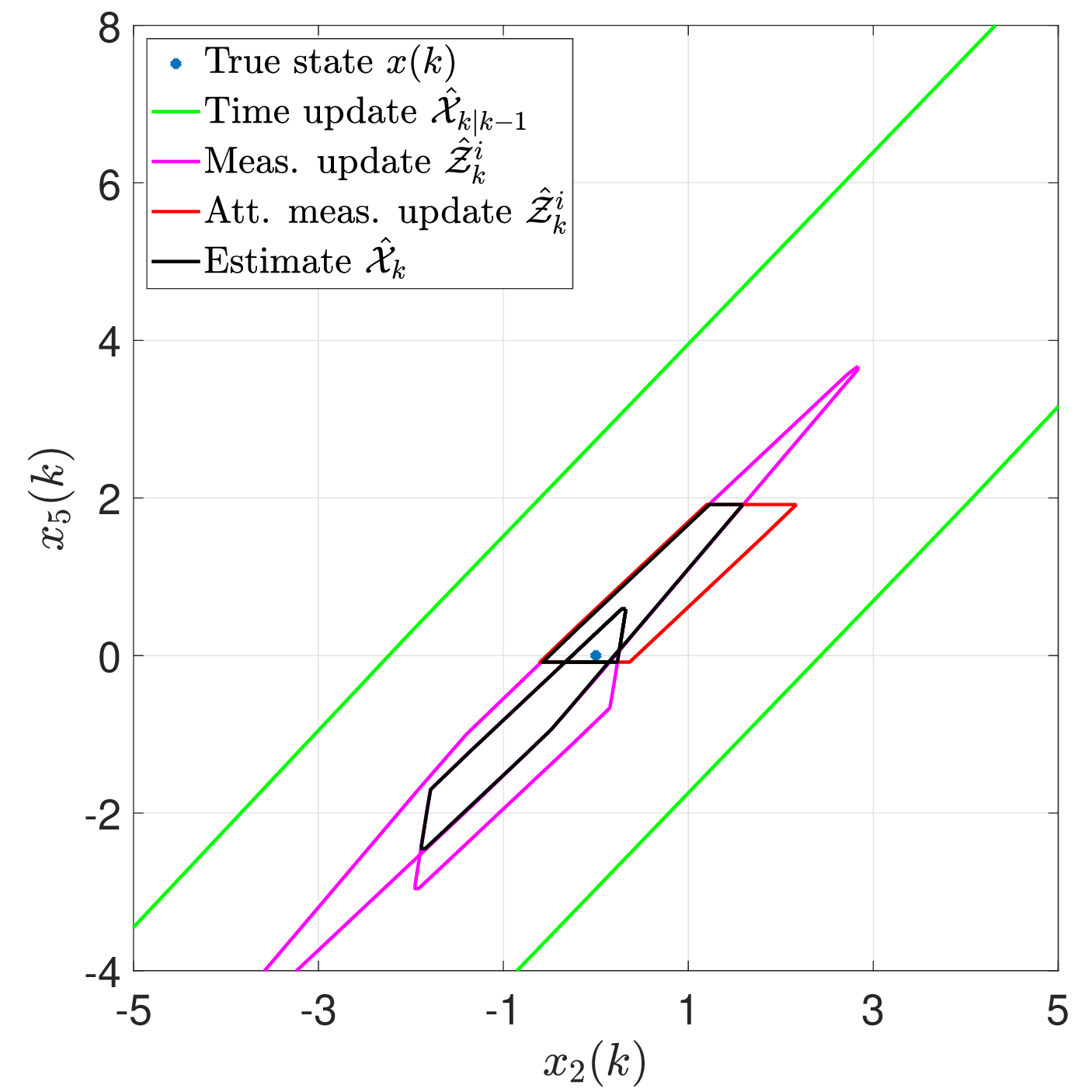}
        \caption{Time $k=1$, Sensor 2 attacked.}
        \label{fig:keq1_s2att}
    \end{subfigure}
    \begin{subfigure}[h]{0.32\textwidth}
     \centering
        \includegraphics[width=\textwidth]{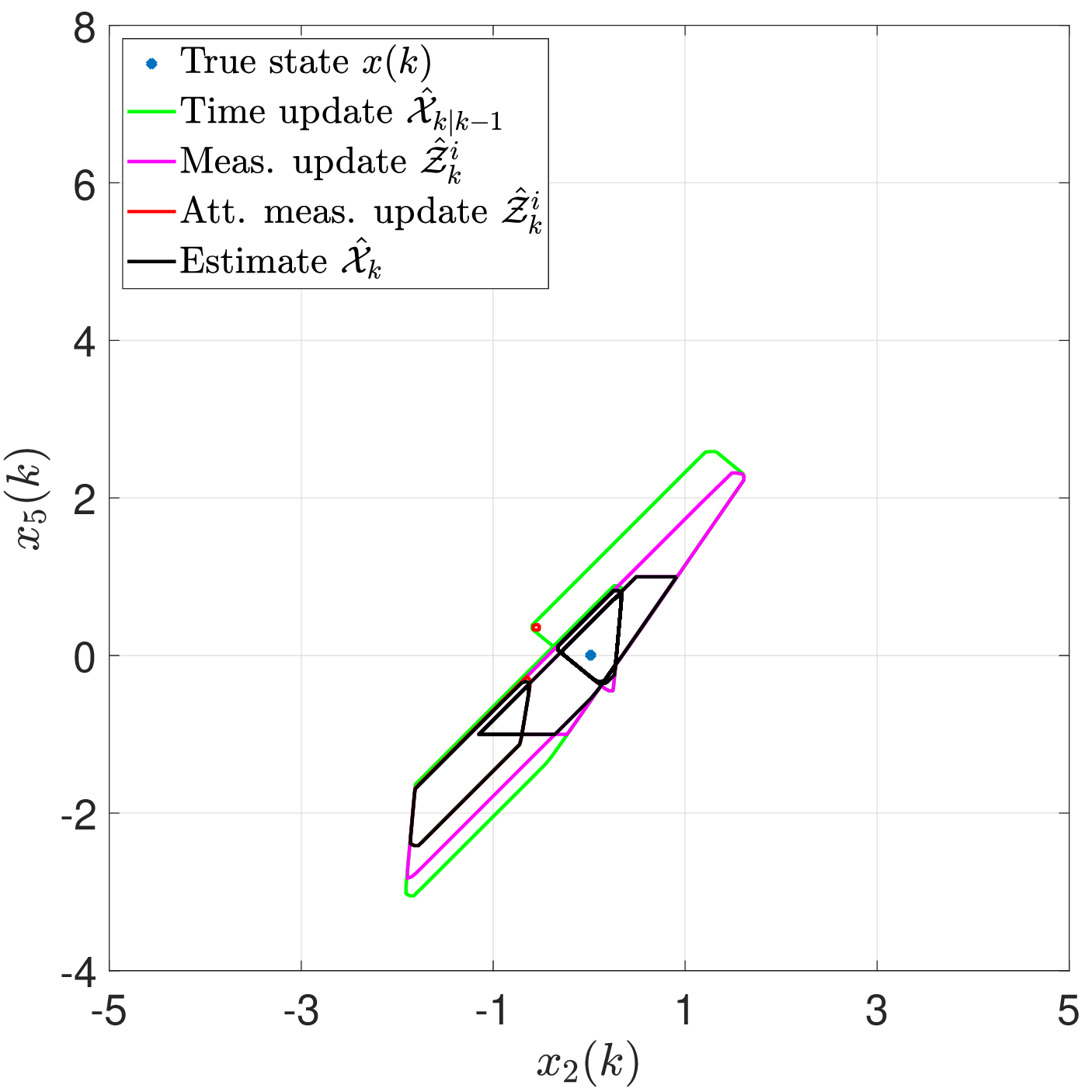}
        \caption{Time $k=2$, Sensor 3 attacked.}
        \label{fig:keq2_s3att}
    \end{subfigure}
    \begin{subfigure}[h]{0.32\textwidth}
     \centering
        \includegraphics[width=\textwidth]{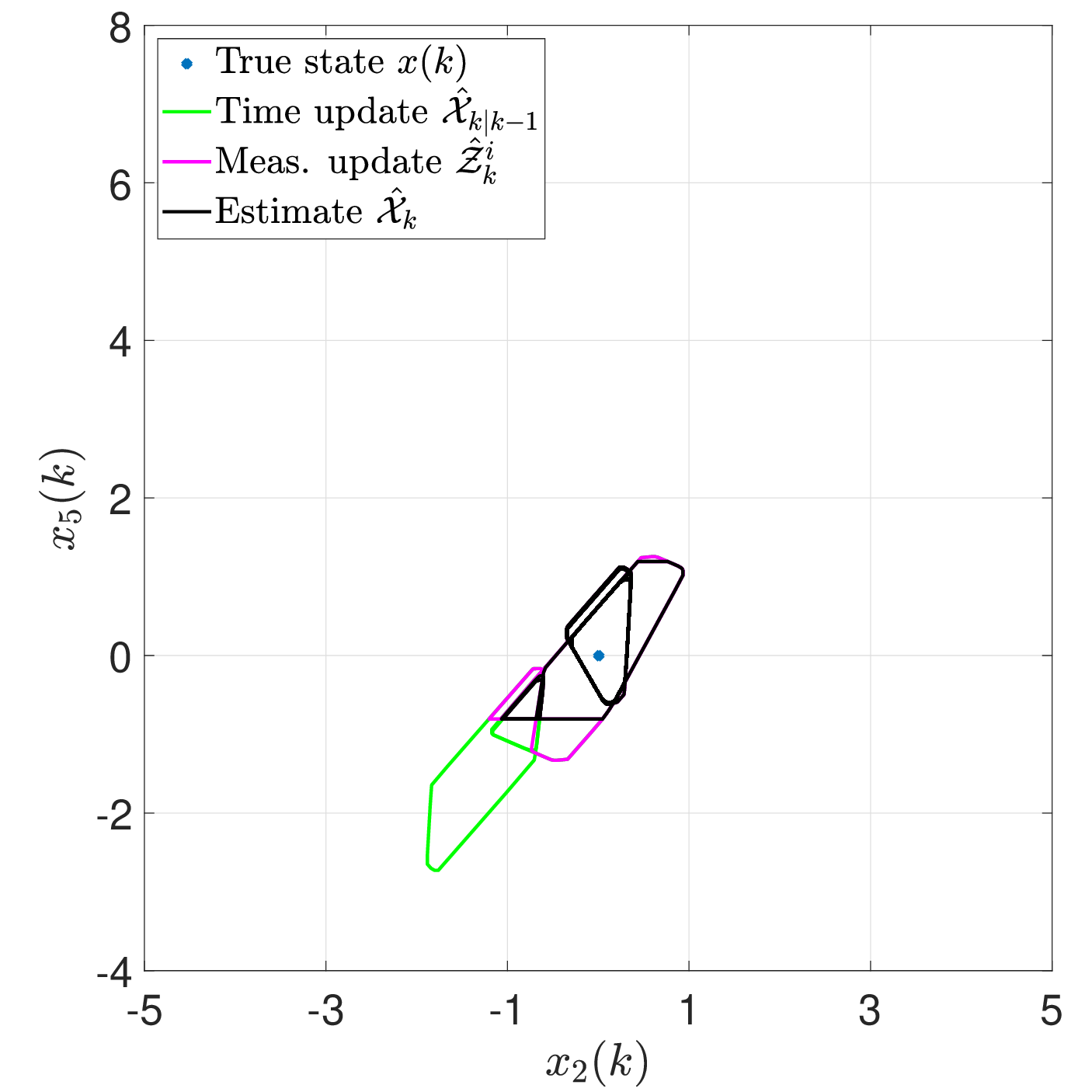}
        \caption{Time $k=3$, Sensor 1 attacked.}
        \label{fig:keq3_s1att}
    \end{subfigure}
\caption{Snapshots of estimated sets using Algorithm~\ref{alg:s3e}.}
    \label{fig:buildingExample}
\end{figure*}

We assume that each floor of the building is equipped with a sensor, i.e., $p=3$, that measures the relative displacement and the velocity of that floor, which can be collected in the output vector $y_i(k)\in\mathbb{R}^3$ as given by \eqref{eq:system_output}, for $i\in\mathbb{Z}_{[1,3]}$, where
\[
\begin{array}{c}
    C_1 = \left[\begin{array}{cccccc}
        1 & -1 & 0 & 0 & 0 & 0 \\
        1 & 0 & -1 & 0 & 0 & 0 \\
        0 & 0 & 0 & 1 & 0 & 0
    \end{array}\right], \quad 
    C_2 = \left[\begin{array}{cccccc}
        -1 & 1 & 0 & 0 & 0 & 0 \\
        0 & 1 & -1 & 0 & 0 & 0 \\
        0 & 0 & 0 & 0 & 1 & 0
    \end{array}\right] \\ [.75cm]
    C_3 = \left[\begin{array}{cccccc}
        -1 & 0 & 1 & 0 & 0 & 0 \\
        0 & -1 & 1 & 0 & 0 & 0 \\
        0 & 0 & 0 & 0 & 0 & 1
    \end{array}\right].
\end{array}
\]
We suppose that the attacker can compromise only one sensor at each time, i.e., $q=1$, whose rationale is explained below.
Although the pairs $(A,C_1)$, $(A,C_2)$, and $(A,C_3)$ are `theoretically' observable, they are not `practically' observable because the corresponding observability matrices have some singular values very close to zero.
Therefore, in this example, we assume redundant observability from every pair of sensors.
Our goal is to monitor the building's floor displacements and velocities irrespective of the compromised sensor.

Suppose the dynamics are corrupted by the process noise $w(k)$, which is bounded by $\mathcal{W}=\zono{0,\sigma_\mathcal{W} I_6}$ with $\sigma_\mathcal{W} =0.02$. 
Similarly, the sensor measurements are corrupted by noise $v_1(k)$, $v_2(k)$, and $v_3(k)$, respectively, which are bounded by $\mathcal{V}_i= \zono{0,\sigma_\mathcal{V} I_3}$ with $\sigma_\mathcal{V} =1$.

To illustrate the efficacy of our algorithm, we apply a similar attack as in \eqref{eq:attgen} to the sensors in which the attacker randomly chooses a sensor~$i$ at every time step~$k$ and injects false data into its measurement $y_i(k)$. In Fig.~\ref{fig:keq1_s2att}, sensor~$2$ is compromised. We compute the estimated set (black), which contains the true state. Then, sensor~$3$ is attacked in Fig.~\ref{fig:keq2_s3att}, in which the time update set $\hat{\mathcal{X}}_{k|k-1}$ shrinks due to the progressive intersection. Finally, sensor~$1$ is attacked in Fig.~\ref{fig:keq3_s1att}. The true state $x(k)$ remains enclosed by the estimated measurement update sets at every time step. Also, the estimation error remains bounded, and the attacker cannot destroy the accuracy of the set-based state estimate. Finally, we remark that the intersections depicted in the figure may not appear accurate because the figure shows a projection of six-dimensional sets on a two-dimensional plane.

\section{Conclusion and Future Work}
\label{sec:conclude}

We presented a secure set-based state estimation algorithm that ensures the inclusion of the true state of an LTI system in the estimated set even when an attacker compromises all but one sensor. 
We achieved this by constructing agreement sets from the intersection of various observable combinations of measurement update sets. We showed that our algorithm guarantees the inclusion of the true state in the estimated set. Moreover, we proposed sufficient conditions for detecting, identifying, and filtering the attack signals and presented a simple algorithm to identify the set of compromised sensors at every time instant. The proposed algorithm may find applications in the safety verification of safety-critical systems whose multiple sensors could be compromised by an attacker.

While our algorithm's worst-case complexity may increase over time under intelligently designed stealthy attacks, we argued that such attacks are challenging to execute due to the requirement of a complete understanding of the system and estimation algorithm, and substantial computational resources. Nonetheless, we suggested various strategies to reduce the complexity of our algorithm to facilitate its implementation. 

Our future work will focus on the set-based secure state estimation of nonlinear systems and developing a data-driven approach for secure estimation when the system model is unknown. We also highlight several open questions that remain unaddressed in this paper. For instance, devising better agreement protocols that can improve the lower bounds on the detectable, identifiable, and filterable attack signals will further strengthen the set-based estimation approach presented in this paper. The question of computational complexity under stealthy attacks remains to be addressed rigorously. The stability of the estimated set is only addressed for the case when strictly less than half of the sensors are attacked. When more than half of the sensors are attacked, it is challenging to upper-bound the estimated set by a contracting set. Finally, from the attacker's perspective, analyzing the complexity of generating stealthy attacks that increase the complexity of the set-based estimation is an interesting research problem.



\bibliographystyle{apalike}
\bibliography{biblio}

@article{alamo2005,
  title={Guaranteed state estimation by zonotopes},
  author={Alamo, Teodoro and Bravo, Jos{\'e} Manuel and Camacho, Eduardo F},
  journal={Automatica},
  volume={41},
  number={6},
  pages={1035--1043},
  year={2005},
  publisher={Elsevier}
}

@inproceedings{alanwar2019,
  title={Distributed secure state estimation using diffusion {Kalman} filters and reachability analysis},
  author={Alanwar, Amr and Said, Hazem and Althoff, Matthias},
  booktitle={58th IEEE Conference on Decision and Control (CDC)},
  pages={4133--4139},
  year={2019}
}

@article{alanwar2023,
  title={Distributed set-based observers using diffusion strategies},
  author={Alanwar, Amr and Rath, Jagat Jyoti and Said, Hazem and Johansson, Karl Henrik and Althoff, Matthias},
  journal={Journal of the Franklin Institute},
  volume={360},
  number={10},
  pages={6976--6993},
  year={2023},
  publisher={Elsevier}
}

@inproceedings{althoff2015,
  title={An introduction to {CORA} 2015},
  author={Althoff, Matthias},
  booktitle={ARCH14-15: 1st and 2nd International Workshop on Applied veRification for Continuous and Hybrid Systems},
  editor={Goran Frehse and Matthias Althoff},
  series={EPiC Series in Computing},
  volume={34},
  pages={120--151},
  year={2015},
  publisher={EasyChair}
}

@article{althoff2021,
  title={Comparison of guaranteed state estimators for linear time-invariant systems},
  author={Althoff, Matthias and Rath, Jagat Jyoti},
  journal={Automatica},
  volume={130},
  pages={109662},
  year={2021},
  publisher={Elsevier}
}

@article{althoff2021-2,
  title={Set propagation techniques for reachability analysis},
  author={Althoff, Matthias and Frehse, Goran and Girard, Antoine},
  journal={Annual Review of Control, Robotics, and Autonomous Systems},
  volume={4},
  number={1},
  pages={369--395},
  year={2021},
  publisher={Annual Reviews}
}

@article{blesa2012,
  title={Robust fault detection using polytope-based set-membership consistency test},
  author={Blesa, Joaquim and Puig, Vicen{\c{c}} and Saludes, Jordi},
  journal={IET Control Theory \& Applications},
  volume={6},
  number={12},
  pages={1767--1777},
  year={2012}
}

@inproceedings{bouron2001,
  title={Set-membership non-linear observers with application to vehicle localisation},
  author={Bouron, Pascal and Meizel, Dominique and Bonnifait, Ph},
  booktitle={European Control Conference (ECC)},
  pages={1255--1260},
  year={2001}
}

@article{chen2022,
  title={Outlier-robust set-membership estimation for discrete-time linear systems},
  author={Chen, Aijun and Ngoc Dinh, Thach and Raissi, Tarek and Shen, Yi},
  journal={International Journal of Robust and Nonlinear Control},
  volume={32},
  number={4},
  pages={2313--2329},
  year={2022},
  publisher={Wiley Online Library}
}

@article{chang2018,
  title={Secure estimation based {Kalman} filter for cyber--physical systems against sensor attacks},
  author={Chang, Young Hwan and Hu, Qie and Tomlin, Claire J},
  journal={Automatica},
  volume={95},
  pages={399--412},
  year={2018},
  publisher={Elsevier}
}

@inproceedings{chong2015,
  title={Observability of linear systems under adversarial attacks},
  author={Chong, Michelle S and Wakaiki, Masashi and Hespanha, Joao P},
  booktitle={American Control Conference (ACC)},
  pages={2439--2444},
  year={2015}
}

@inproceedings{chong2020,
  title={A secure state estimation algorithm for nonlinear systems under sensor attacks},
  author={Chong, Michelle S and Sandberg, Henrik and Hespanha, Joao P},
  booktitle={59th IEEE Conference on Decision and Control (CDC)},
  pages={5743--5748},
  year={2020}
}

@article{depaula2022,
  title={Zonotopic Filtering for Uncertain Nonlinear Systems: Fundamentals, Implementation Aspects, and Extensions [Applications of Control]},
  author={de Paula, Alesi A and Raffo, Guilherme V and Teixeira, Bruno OS},
  journal={IEEE Control Systems Magazine},
  volume={42},
  number={1},
  pages={19--51},
  year={2022},
  publisher={IEEE}
}

@article{fawzi2014,
  title={Secure estimation and control for cyber-physical systems under adversarial attacks},
  author={Fawzi, Hamza and Tabuada, Paulo and Diggavi, Suhas},
  journal={IEEE Transactions on Automatic Control},
  volume={59},
  number={6},
  pages={1454--1467},
  year={2014},
  publisher={IEEE}
}

@book{freeman2008,
  title={Robust nonlinear control design: state-space and Lyapunov techniques},
  author={Freeman, Randy and Kokotovic, Petar V},
  year={2008},
  publisher={Springer Birkhäuser: Boston, MA}
}

@article{he2021,
  title={How to secure distributed filters under sensor attacks},
  author={He, Xingkang and Ren, Xiaoqiang and Sandberg, Henrik and Johansson, Karl Henrik},
  journal={IEEE Transactions on Automatic Control},
  volume={67},
  number={6},
  pages={2843--2856},
  year={2021},
  publisher={IEEE}
}

@article{jaulin2009,
  title={Robust set-membership state estimation: Application to underwater robotics},
  author={Jaulin, Luc},
  journal={Automatica},
  volume={45},
  number={1},
  pages={202--206},
  year={2009}
}

@article{kayan2022,
  title={Cybersecurity of industrial cyber-physical systems: A review},
  author={Kayan, Hakan and Nunes, Matthew and Rana, Omer and Burnap, Pete and Perera, Charith},
  journal={ACM Computing Surveys},
  volume={54},
  number={11s},
  pages={1--35},
  year={2022},
  publisher={ACM New York, NY}
}

@article{khazraei2022,
  title={Attack-resilient state estimation with intermittent data authentication},
  author={Khazraei, Amir and Pajic, Miroslav},
  journal={Automatica},
  volume={138},
  pages={110035},
  year={2022},
  publisher={Elsevier}
}

@article{kim2018,
  title={Detection of sensor attack and resilient state estimation for uniformly observable nonlinear systems having redundant sensors},
  author={Kim, Junsoo and Lee, Chanhwa and Shim, Hyungbo and Eun, Yongsoon and Seo, Jin H},
  journal={IEEE Transactions on Automatic Control},
  volume={64},
  number={3},
  pages={1162--1169},
  year={2018},
  publisher={IEEE}
}

@article{le2013,
  title={Zonotopic guaranteed state estimation for uncertain systems},
  author={Le, Vu Tuan Hieu and Stoica, Cristina and Alamo, Teodoro and Camacho, Eduardo F and Dumur, Didier},
  journal={Automatica},
  volume={49},
  number={11},
  pages={3418--3424},
  year={2013},
  publisher={Elsevier}
}

@article{lee2020,
  title={Fully distributed resilient state estimation based on distributed median solver},
  author={Lee, Jin Gyu and Kim, Junsoo and Shim, Hyungbo},
  journal={IEEE Transactions on Automatic Control},
  volume={65},
  number={9},
  pages={3935--3942},
  year={2020},
  publisher={IEEE}
}

@article{lesi2017,
  title={Security-aware scheduling of embedded control tasks},
  author={Lesi, Vuk and Jovanov, Ilija and Pajic, Miroslav},
  journal={ACM Transactions on Embedded Computing Systems},
  volume={16},
  number={5s},
  pages={1--21},
  year={2017},
  publisher={ACM New York, NY, USA}
}

@article{li2021,
  title={Distributed set-membership filtering for discrete-time systems subject to denial-of-service attacks and fading measurements: A zonotopic approach},
  author={Li, Xin and Wei, Guoliang and Wang, Licheng},
  journal={Information Sciences},
  volume={547},
  pages={49--67},
  year={2021},
  publisher={Elsevier}
}

@article{li2023,
  title={Attack detection for cyber-physical systems: A zonotopic approach},
  author={Li, Jitao and Wang, Zhenhua and Shen, Yi and Xie, Lihua},
  journal={IEEE Transactions on Automatic Control},
  volume={68},
  number={11},
  pages={6828--6835},
  year={2023},
  publisher={IEEE}
}

@article{liu2020,
  title={Distributed set-membership filtering for time-varying systems under constrained measurements and replay attacks},
  author={Liu, Lei and Ma, Lifeng and Wang, Yiwen and Zhang, Jie and Bo, Yuming},
  journal={Journal of the Franklin Institute},
  volume={357},
  number={8},
  pages={4983--5003},
  year={2020},
  publisher={Elsevier}
}

@article{liu2021,
  title={Distributed non-fragile set-membership filtering for nonlinear systems under fading channels and bias injection attacks},
  author={Liu, Lei and Ma, Lifeng and Zhang, Jie and Bo, Yuming},
  journal={International Journal of Systems Science},
  volume={52},
  number={6},
  pages={1192--1205},
  year={2021},
  publisher={Taylor \& Francis}
}

@article{meslem2020,
  title={Robust set-membership state estimator against outliers in data},
  author={Meslem, Nacim and Hably, Ahmad},
  journal={IET Control Theory \& Applications},
  volume={14},
  number={13},
  pages={1752--1761},
  year={2020},
  publisher={Wiley Online Library}
}

@article{mitra2019,
  title={Byzantine-resilient distributed observers for {LTI} systems},
  author={Mitra, Aritra and Sundaram, Shreyas},
  journal={Automatica},
  volume={108},
  pages={108487},
  year={2019},
  publisher={Elsevier}
}

@article{niazi2023,
  title={Resilient set-based state estimation for linear time-invariant systems using zonotopes},
  author={Niazi, M Umar B and Alanwar, Amr and Chong, Michelle S and Johansson, Karl H},
  journal={European Journal of Control},
  volume={74},
  pages={100837},
  year={2023},
  publisher={Elsevier}
}

@article{pajic2016,
  title={Attack-resilient state estimation for noisy dynamical systems},
  author={Pajic, Miroslav and Lee, Insup and Pappas, George J},
  journal={IEEE Transactions on Control of Network Systems},
  volume={4},
  number={1},
  pages={82--92},
  year={2016},
  publisher={IEEE}
}

@article{rego2020,
  title={Guaranteed methods based on constrained zonotopes for set-valued state estimation of nonlinear discrete-time systems},
  author={Rego, Brenner S and Raffo, Guilherme V and Scott, Joseph K and Raimondo, Davide M},
  journal={Automatica},
  volume={111},
  pages={108614},
  year={2020},
  publisher={Elsevier}
}

@article{rego2021,
  title={State estimation and leakage detection in water distribution networks using constrained zonotopes},
  author={Rego, Brenner S and Vrachimis, Stelios G and Polycarpou, Marios M and Raffo, Guilherme V and Raimondo, Davide M},
  journal={IEEE Transactions on Control Systems Technology},
  volume={30},
  number={5},
  pages={1920--1933},
  year={2021},
  publisher={IEEE}
}

@article{scott2016,
  title={Constrained zonotopes: A new tool for set-based estimation and fault detection},
  author={Scott, Joseph K and Raimondo, Davide M and Marseglia, Giuseppe Roberto and Braatz, Richard D},
  journal={Automatica},
  volume={69},
  pages={126--136},
  year={2016},
  publisher={Elsevier}
}

@article{segovia2024,
  title={A survey on cyber-resilience approaches for cyber-physical systems},
  author={Segovia-Ferreira, Mariana and Rubio-Hernan, Jose and Cavalli, Ana and Garcia-Alfaro, Joaquin},
  journal={ACM Computing Surveys},
  volume={56},
  number={8},
  pages={1--37},
  year={2024},
  publisher={ACM New York, NY}
}

@article{shinohara2018b,
  title={Reach Set-Based Secure State Estimation against Sensor Attacks with Interval Hull Approximation},
  author={Shinohara, Takumi and Namerikawa, Toru},
  journal={SICE Journal of Control, Measurement, and System Integration},
  volume={11},
  number={5},
  pages={399--408},
  year={2018},
  publisher={Taylor \& Francis}
}

@article{shoukry2017,
  title={Secure state estimation for cyber-physical systems under sensor attacks: A satisfiability modulo theory approach},
  author={Shoukry, Yasser and Nuzzo, Pierluigi and Puggelli, Alberto and Sangiovanni-Vincentelli, Alberto L and Seshia, Sanjit A and Tabuada, Paulo},
  journal={IEEE Transactions on Automatic Control},
  volume={62},
  number={10},
  pages={4917--4932},
  year={2017},
  publisher={IEEE}
}

@article{shoukry2018,
  title={{SMT}-based observer design for cyber-physical systems under sensor attacks},
  author={Shoukry, Yasser and Chong, Michelle and Wakaiki, Masashi and Nuzzo, Pierluigi and Sangiovanni-Vincentelli, Alberto and Seshia, Sanjit A and Hespanha, Joao P and Tabuada, Paulo},
  journal={ACM Transactions on Cyber-Physical Systems},
  volume={2},
  number={1},
  pages={1--27},
  year={2018},
  publisher={ACM New York, NY, USA}
}

@article{song2019,
  title={Set-membership estimation for complex networks subject to linear and nonlinear bounded attacks},
  author={Song, Haiyu and Shi, Peng and Lim, Cheng-Chew and Zhang, Wen-An and Yu, Li},
  journal={IEEE Transactions on Neural Networks and Learning Systems},
  volume={31},
  number={1},
  pages={163--173},
  year={2019},
  publisher={IEEE}
}

@article{truong2021,
  title={Analysis of networked structural control with packet loss},
  author={Truong, Thao HT and Seiler, Peter and Linderman, Lauren E},
  journal={IEEE Transactions on Control Systems Technology},
  volume={30},
  number={1},
  pages={344--351},
  year={2021},
  publisher={IEEE}
}

@article{yang2018comparison,
  title={A comparison of zonotope order reduction techniques},
  author={Yang, Xuejiao and Scott, Joseph K},
  journal={Automatica},
  volume={95},
  pages={378--384},
  year={2018},
  publisher={Elsevier}
}

@article{yong2018,
  title={Switching and data injection attacks on stochastic cyber-physical systems: Modeling, resilient estimation, and attack mitigation},
  author={Yong, Sze Zheng and Zhu, Minghui and Frazzoli, Emilio},
  journal={ACM Transactions on Cyber-Physical Systems},
  volume={2},
  number={2},
  pages={1--2},
  year={2018},
  publisher={ACM New York, NY, USA}
}

@article{zhang2020,
  title={A novel set-membership estimation approach for preserving security in networked control systems under deception attacks},
  author={Zhang, Yilian and Zhu, Yanfei and Fan, Qinqin},
  journal={Neurocomputing},
  volume={400},
  pages={440--449},
  year={2020},
  publisher={Elsevier}
}

@book{zhou1998,
  title={Essentials of Robust Control},
  author={Zhou, Kemin and Doyle, John Comstock},
  volume={104},
  year={1998},
  publisher={Prentice Hall: Upper Saddle River, NJ}
}

@article{zhu2022,
  title={Consensus and Security Control of Multi-agent Systems Based on Set-membership Estimation with Time-varying Topology under Deception Attacks},
  author={Zhu, Yanfei and Liu, Hang and Li, Chuanjiang and Yu, Jiahao},
  journal={International Journal of Control, Automation and Systems},
  volume={20},
  number={11},
  pages={3624--3636},
  year={2022},
  publisher={Springer}
}

\end{document}